\documentclass{amsart}
\usepackage[dvips]{graphicx}
\usepackage{amscd}
\usepackage{amsxtra}
\usepackage{enumerate}
\usepackage{amsfonts}
\usepackage{amsmath}
\usepackage{amssymb}
\usepackage{amsthm}

\newtheorem{theorem}{Theorem}[section]
\newtheorem{corollary}[theorem]{Corollary}
\newtheorem{lemma}[theorem]{Lemma}
\newtheorem{proposition}[theorem]{Proposition}
\theoremstyle{definition}
\newtheorem{definition}[theorem]{Definition}

\newtheorem{example}[theorem]{Example}
\newtheorem{claim}[theorem]{Claim}

\theoremstyle{remark}
\newtheorem{remark}[theorem]{Remark}
\newtheorem*{acknowledgements}{Acknowledgements}

\numberwithin{equation}{section}

\usepackage{graphicx}

\newcommand{\Iidentity}{\hbox{\upshape \small1\kern-3.3pt\normalsize1}}
\def\blfootnote{\xdef\@thefnmark{}\@footnotetext}

\begin{document}

\title[Spectral Theory of Discrete Processes]{Spectral Theory of Discrete Processes}

\author{Palle E. T. Jorgensen}
\address{Department of Mathematics, The University of Iowa, Iowa City, IA52242, USA}
\email{jorgen@math.uiowa.edu}
\urladdr{http://www.math.uiowa.edu/~jorgen}
\thanks{Work supported in part by the U.S. National Science Foundation}

\author{Myung-Sin Song}
\address{Department of Mathematics and Statistics, Southern Illinois University Edwardsville, Edwardsville, IL62026, USA}
\curraddr{}
\email{msong@siue.edu}
\urladdr{http://www.siue.edu/~msong}

\subjclass{Primary 42C40, 47S50, 62B5, 68U10, 94A08}
\date{Mar, 2009}


\keywords{Hilbert space, spectrum, encoding, transfer operator, infinite graphs, wavelets, fractals, stochastic process, path measures, transition probability.}



\begin{abstract}
We offer a spectral analysis for a class of transfer operators. These 
transfer operators arise for a wide range of stochastic processes, ranging 
from random walks on infinite graphs to the processes that govern signals 
and recursive wavelet algorithms; even spectral theory for fractal measures. 
In each case, there is an associated class of harmonic functions which we 
study. And in addition, we study three questions in depth:

In specific applications, and for a specific stochastic process, how do we 
realize the transfer operator $T$ as an operator in a suitable Hilbert space? 
And how to spectral analyze $T$ once the right Hilbert space $\mathcal{H}$ 
has been selected? Finally we characterize the stochastic processes that are 
governed by a single transfer operator.

In our applications, the particular stochastic process will live on an 
infinite path-space which is realized in turn on a state space $S$. In the 
case of random walk on graphs $G$, $S$ will be the set of vertices of $G$. 
The Hilbert space $\mathcal{H}$ on which the transfer operator $T$ acts will 
then be an $L^{2}$ space on $S$, or a Hilbert space defined from an 
energy-quadratic form. 

This circle of problems is both interesting and non-trivial as it turns out 
that $T$ may often be an unbounded linear operator in $\mathcal{H}$; but even 
if it is bounded, it is a non-normal operator, so its spectral theory is not 
amenable to an analysis with the use of von Neumann's spectral theorem. 
While we offer a number of applications, we believe that our spectral 
analysis will have intrinsic interest for the theory of operators in Hilbert 
space. 
\end{abstract}

\maketitle \tableofcontents

\section{Introduction}
\label{sec:1}
In this paper, we consider infinite configurations of vectors $(f_{k})_{k \in \mathbb{Z}}$ 
in a Hilbert space $\mathcal{H}$. Since our Hilbert spaces $\mathcal{H}$ are 
typically infinite-dimensional, this can be quite complicated, and it will be 
difficult to make sense of finite and infinite linear combinations 
$\sum_{k \in \mathbb{Z}}c_{k}f_{k}$.  

In case the system $(f_{k})$ is orthogonal, the problem is easy, but 
non-orthogonality serves as an encoding of statistical correlations, which 
in turn motivates our study.  In applications, a 
particular system of vectors $f_{k}$ may often be analyzed with the use of 
a single unitary operator $U$ in $\mathcal{H}$.  This happens if there is a 
fixed vector $\varphi \in \mathcal{H}$ such that $f_{k} = U^{k}\varphi$ 
for all $k \in \mathbb{Z}$.  When this is possible, the spectral theorem 
will then apply to this unitary operator.  A key idea in our paper is to 
identify a spectral density function and a transfer operator, both computed 
directly from the pair $(\varphi, U)$.  

We show that the study of linear expressions $\sum_{k}c_{k}f_{k}$ may be done 
with the aid of the spectral function for a pair $(\varphi, U)$.  A 
spectral function for a unitary operator $U$ is really a system of 
functions $(p_{\varphi})$, one for each cyclic subspace 
$\mathcal{H}(\varphi)$.  In each cyclic subspace, the function $p_{\varphi}$ 
is a complete unitary invariant for $U$ restricted to 
$\mathcal{H}(\varphi)$: by this we mean that the function $p_{\varphi}$ 
encodes all the spectral data coming from the vectors $f_{k}=U^{k}\varphi$, 
$k \in \mathbb{Z}$. For background literature on spectral function and 
their applications we refer to \cite{BH05, JP05, LWW04, PSWX04, Sad06, TT07}. 

In summary, the spectral representation theorem is the assertion 
that commuting unitary operators in Hilbert space may be represented as 
multiplication operators in an $L^{2}$-Hilbert space. The understanding 
is that this representation is defined as a unitary equivalence, and 
that the $L^{2}$-Hilbert space to be used allows arbitrary measures, and 
$L^{2}$ will be a Hilbert space of vector valued functions, 
see e.g., \cite{Hel86}. Because of applications, our systems of vectors will
be indexed by an arbitrary discrete set rather than merely integers 
$\mathbb{Z}$.  

We will attack this problem via an isometric embedding of $\mathcal{H}$ into 
and $L^{2}$-space built on infinite parths in such a way that the vectors 
$f_{k}$ in $\mathcal{H}$ transform into a system of random variables $Z_{k}$.  
Specifically, via certain encodings we build a path-space $\Omega$ for the 
particular problem at hand as well as a path space measure $\mathbb{P}$ 
defined on a $\sigma$-algebra of subsets of $\Omega$.

If $\mathcal{H}$ consists of a space of functions $f$ on a state space $S$, 
we will need the covariance numbers
\[
  \mathbb{E}((f_{1}\circ Z_{n})\cdot(f_{2}\circ Z_{m})):=
  \int_{\Omega}f_{1}(Z_{n}(\gamma))f_{2}(Z_{m}(\gamma))d\mathbb{P}(\gamma),
\]
where $Z_{n}:\Omega \to S$, i.e., where the stochastic process is $S$-valued. 
The set $S$ is called the state space.

The paper is organized as follows. In section \ref{sec:SP}, for later use, 
we present our path-space approach, and we discuss the path-space measures 
that we will use in computing transitions for stochastic processes. We prove 
two theorems making the connection between our path-space measures on the 
one hand, and the operator theory on the other. Several preliminary results 
are established proving how the transfer operator governs the process and 
its applications.

The applications we give in sections \ref{sec:G} and \ref{sec:STO} are 
related. In fact, we unify these applications with the use of an encoding 
map which is also studied in detail. It is applied to transitions on 
certain infinite graphs, to dynamics of (non-invertible) endomorphisms 
(measures on solenoids), to digital filters and their use in wavelets and 
signals, and to harmonic analysis on fractals.

The remaining sections deal primarily with applications to a sample of 
concrete cases.

\section{Stochastic Processes}
\label{sec:SP}
A key tool in our analysis is the construction of path-space measures on 
infinite paths, primarily in the case of discrete paths, but the fundamental 
ideas are the same in the continuous case. Both viewpoints are used in 
\cite{JoPe08}. Readers who wish to review the ideas behind there 
constructions (stochastic processes and consistent families of measures) are 
referred to \cite{Hi80, Hi93, Hi94} and \cite{Ne73}. 

Let $(\Omega, \mathcal{F}, \mathbb{P})$ be a Borel probablity space, 
$\Omega$ compact Hausdorff space. (Expectation 
$\mathbb{E}(\cdot)=\int_{\Omega}\cdot d\mathbb{P}$.)

Let $(Z_{k})_{k \geq 0}$ be a stochastic process, and 
\begin{equation}
\label{E:SP.1}
  \mathcal{F}_{n}=\sigma \text{-alg.}\{Z_{k} | k \leq n\} 
\end{equation}
the corresponding filtration.  Let 
$\mathcal{A}_{n}:=$ the subspace in $L^{2}(\Omega, \mathbb{P})$ generated by 
$\mathcal{F}_{n}$.  Let $P_{n}$ be the orthogonal projection of 
$L^{2}(\Omega, \mathbb{P})$ onto $\mathcal{A}_{n}$; then the conditional 
expectations $\mathbb{E}(\cdot | \mathcal{F}_{n})$ is simply $=P_{n}$.

We say that $(Z_{k})_{k \geq 0}$ has the generalized Markov property if and 
only if there exists a state space $S$ (also a compact Borel space):
\[
  Z_{k}: \Omega \to S 
\]
such that for all bounded functions $f$ on $S$, for all 
$n \in \mathbb{N}_{\geq 0}$, 
$\mathbb{E}(f|\mathcal{F}_{n})=\mathbb{E}(f|Z_{n})$.

To make precise the operator theoretic tools going into our construction, 
we must first introduce the ambient Hilbert spaces. We are restricting here 
to $L^{2}$ processes, so the corresponding stochastic integrals will take 
values in an ambient $L^{2}$-space of random variables: For our analysis, we 
must therefore specify a fixed probability space, with $\sigma$-algebra and 
probability measure.

We will have occasion to vary this initial probability space, depending on 
the particular transition operator that governs the process.

In the most familiar case of Brownian motion, or random walk, the 
probability space amounts to a somewhat standard construction of Wiener and 
Kolmogorov, but here with some modification for our problem at hand: The
essential axiom in Wiener's case is that all finite samples are jointly 
Gaussian, but we will drop this restriction and consider general stochastic 
processes, and so we will not make restricting assumptions on the sample 
distributions and on the underlying probability space. For more details, and 
concrete applications, regarding this stochastic approach and its 
applications, see sections \ref{sec:SP} and \ref{sec:STO} below.

We begin here with a particular case of a process taking values in the set of 
vertices in a fixed infinite graph G: \cite{JoSo09}

\subsection{Starting Assumptions and Constructions.}
\label{sec:SP.1}
\begin{enumerate}[(a) ]
  \item $G=(G_{0}, G_{1})$ a graph, $G_{0}=$ the set of vertices, $G^{1}=$ 
  the set of edges.
  \item $(S, \mathcal{B}_{S}, \mu)$ a probability space.
  \item The \textit{transition matrix} is the function 
    \[
      p(x,y):=\mathbb{P}(\{\gamma\in\Omega |Z_{n}(\gamma)=x, 
      Z_{n+1}(\gamma)=y\})
    \] 
    defined for all $(x,y)\in G^{1}$, and we assume that it is 
    independent of $n$. 
  \item From (a) and (b), we construct the path space 
    \[
      \Omega:=\{\gamma=(x_{0}x_{1}x_{2}\cdots)|(x_{i-1}x_{i})\in G^{1}, 
      \forall i \in \mathbb{N}\},
    \] 
    and the path-measure 
    $\mathbb{P}=\mathbb{P}_{\mu}$.  The cylinder sets given by the following 
    data: For 
    $E_{i}\in\mathcal{B}_{S}$, $E_{i}\subset S$, set 
    \begin{align*}
      &\mathbb{P}(C(E_{1}, \cdots, E_{n})) \\
      &:=\int_{E_{0}}\int_{E_{1}}\cdots \int_{E_{n}}p(x_{0},x_{1})p(x_{1},x_{2})
      \cdots p(x_{n-1},x_{n})d\mu(x_{0})d\mu(x_{1})\cdots d\mu(x_{n}) 
    \end{align*}
  \item Starting with $(\Omega, \mathcal{F}, \mathbb{P})$, if 
    $\mathcal{G} \subset \mathcal{F}$ is a subsigma algebra, let 
    $\mathbb{E}(\cdot | \mathcal{G})$ be the conditional expectation, 
    conditioned by $\mathcal{G}$. \\
    If $(X_{i})$ is a family of random variables, and $\mathcal{G}$ is the 
    $\sigma$-algebra generated by $(X_{i})$ we write 
    $\mathbb{E}(\cdot |(X_{i}))$ in place of $\mathbb{E}(\cdot | \mathcal{G})$.
  \item Let $(\Omega, \mathcal{F}, \mathbb{P}, (Z_{n}))$ be as above. We
    say that $(Z_{n})$ is \textit{Markov} if and only if 
    \[
      \mathbb{E}(f\circ Z_{n+1}|\{Z_{0},\cdots Z_{n}\})
      =\mathbb{E}(f\circ Z_{n+1}|Z_{n}) \quad 
      \mbox{for all $n\in \mathbb{N}_{0}$}.
    \]
  \item From (b) and (d) we define the \textit{transfer operator} $T$ by
    \begin{equation}
    \label{E:SP.0.1}
      (Tf)(x)=\int_{S}p(x,y)f(y)d\mu(y)
    \end{equation}
    for measurable functions $f$ on $S$.
    If $\Iidentity$ denote the constant function $1$ on $S$, then 
    $T\Iidentity=\Iidentity$.
  \item Let $(S, \mathcal{B}_{S}, \mu)$ and $T$ be as in (g), 
    see(\ref{E:SP.0.1}). 
    A measure $\mu_{0}$ on $S$ is said to be a 
    \textit{Perron-Frobenius measure} if and only if 
    \begin{equation}
    \label{E:SP.0.2}
      \int_{S}(Tf)(x)d\mu_{0}(x)=\int_{S}f(x)d\mu_{0}(x), \quad 
      \mbox{abbreviated $\mu_{0}\circ T=\mu_{0}$}.
    \end{equation}
  \item Let $(\Omega, \mathcal{F}, \mathbb{P})$ be as above, and let $T$ be 
    the transfer operator.  If $\mu_{0}$ is a Perron-Frobenius measure, let 
    $\mathbb{P}^{(\mu_{0})}$ be the measure on $\Omega$ determined by using 
    $\mu_{0}$ as the first factor, i.e., 
    \begin{align*}
      &\mathbb{P}^{(\mu_{0})}(C(E_{1}, \cdots, E_{n})) \\
      &=\int_{E_{0}}\int_{E_{1}}\cdots \int_{E_{n}}p(x_{0},x_{1})p(x_{1},x_{2})
      \cdots p(x_{n-1},x_{n})d\mu_{0}(x_{0})d\mu(x_{1})\cdots d\mu(x_{n}) \\
      &=\int_{E_{0}}\mathbb{P}_{x_{0}}(C(E_{1}, \cdots, E_{n}))d\mu(x_{0}). 
    \end{align*}
\end{enumerate}

In many cases, it is possible to choose specific Perron-Frobenius measures 
$\mu_{0}$, i.e., measures $\mu_{0}$ satisfying
\[
  \mu_{0}(S)=1  \mbox{ and } 
  \int_{S}(Tf)(x)d\mu_{0}(x)=\int_{S}f(x)d\mu_{0}(x).
\]
(Note the normalization!)

\begin{theorem}
\label{T:SP.0.1}
(D. Ruelle) \cite{Ba00} Suppose there is a norm $\|\cdot \|$ on bounded 
measurable functions $f$ on $S$ such that the $\|\cdot \|$-completion 
$L(S)$ is embedded in $L^{\infty}(S)$, and that there are constants 
$\alpha \in (0,1)$, $M \in \mathbb{R}_{+}$ such that 
\[
  \|Tf\| \leq \alpha\|f\|+M\|f\|_{\infty},
\]
where $\|\cdot\|_{\infty}$ is the 
essential supremum-norm.  Then $T$ has a Perron-Frobenius measure.
\end{theorem}

\begin{theorem}
\label{T:SP.0.2}
Let $(S, \mu)$ be a probability space with $S$ carrying a separate 
$\sigma$-algebra $\mathcal{B}_{S}$ and $\mu$ defined on $\mathcal{B}_{S}$. 
Let $\Omega$ be the path space, and supposed the transfer operator $T$ has 
a Perron-Frobenius measure $\mu_{0}$, then 
\begin{equation}
\label{E:SP.1.1}
  \mathbb{E}^{(\mu_{0})}((\overline{\varphi}\circ Z_{n})(\psi \circ Z_{n+1}))
  =\langle \varphi, T\psi \rangle_{L^{2}(\mu_{0})}
\end{equation}
for all $\varphi, \psi \in L^{2}(\mu)$, and all $n \in \mathbb{N}_{0}$.  
Here $\mathbb{E}(F):=\int_{\Omega}F(\omega)d\mathbb{P}(\omega)$ for all 
integrable random variables $F:\Omega \to \mathbb{C}$; $\mathbb{E}$ for 
expectation.
\end{theorem}
\begin{proof}
\begin{align*}
  &\mathbb{E}^{(\mu_{0})}((\overline{\varphi}\circ Z_{n})(\psi 
  \circ Z_{n+1})) \\
  &=\int_{S}\int_{S}\cdots \int_{S}p(x_{0},x_{1})\cdots p(x_{n-1},x_{n})
    p(x_{n},x_{n+1})\overline{\varphi}(x_{n})\psi(x_{n+1})d\mu_{0}(x_{0})
    d\mu(x_{1})\cdots d\mu(x_{n+1}) \\
  &=\int_{S}\int_{S}\cdots \int_{S}p(x_{0},x_{1})\cdots p(x_{n-1},x_{n})
    \overline{\varphi}(x_{n})(T\psi)(x_{n})d\mu_{0}(x_{0})
    d\mu(x_{1})\cdots d\mu(x_{n}) \\
  &=\int_{S}T^{n}(\overline{\varphi}\cdot(T\psi))(x_{0})d\mu_{0}(x_{0}) \\
  &=\int_{S}\overline{\varphi}(x)(T\psi)(x)d\mu_{0}(x) 
    \quad \mbox{ by Perron-Frobenius}\\
  &=\langle \varphi, T\psi \rangle_{L^{2}(\mu_{0})}.
\end{align*}
\end{proof}

It is not necessary in (\ref{E:SP.1.1}) to restrict attention to functions 
$\varphi, \psi$ in $L^{2}(\mu_{0})$.  The important thing is that the integral 
$\int_{S}\overline{\varphi(x)}(T\psi)(x)d\mu_{0}(x)$ exists, and this quantity 
may then be used instead on the RHS in (\ref{E:SP.1.1}).

Let $(Z_{n})_{n\in \mathbb{N}_{0}}$ be a stochastic process, and let 
$\mathcal{F}_{n}$ be the $\sigma$-algebra generated by 
$\{Z_{k}| \text{ } 0 \leq k \leq n\}$.  Futhermore, let 
$\mathbb{E}(\cdot |\mathcal{F}_{n})$ be the conditioned expectation 
conditioned by $\mathcal{F}_{n}$.

\begin{theorem}
\label{T:SP.1}
Let $(Z_{n})_{n\in \mathbb{N}_{0}}$ be a stochastic process with 
stationary transitions and operator $T$.  Then 
\begin{equation}
\label{E:SP.2}
  \mathbb{E}(f\circ Z_{n+1} | \mathcal{F}_{n})=(Tf)\circ Z_{n}
\end{equation}
for all bounded measurable functions $f$ on $S$, and all $n\in \mathbb{N}_{0}$
\end{theorem}
\begin{proof}
We may assume that $f$ is a real valued function on $S$. Let 
$\mathcal{A}_{n}:=$ all bounded $\mathcal{F}_{n}$-measurable functions.  Then
the assertion in (\ref{E:SP.2}) may be restated as:
\begin{equation}
\label{E:SP.3}
  \int_{\Omega}\varphi (f\circ Z_{n+1})d\mathbb{P} 
  = \int_{\Omega}\varphi ((Tf)\circ Z_{n}) d\mathbb{P}
\end{equation}
for all $\varphi \in \mathcal{A}_{n}$.

If $\varphi \in \mathcal{A}_{n}$, 
$\varphi(\cdot)=\Phi(x_{0},x_{1}, \cdots x_{n})$; and then the LHS in 
(\ref{E:SP.3}) may be written as 
\begin{align*}
  &\int_{S}\int_{S}\cdots \int_{S}p(x_{0},x_{1})\cdots p(x_{n},x_{n+1})
    \Phi(x_{0},x_{1}\cdots x_{n})f(x_{n+1})
    d\mu_{0}(x_{0})d\mu(x_{1})\cdots d\mu(x_{n+1}) \\
  &=\int_{S}\int_{S}\cdots \int_{S}p(x_{0},x_{1})\cdots p(x_{n-1},x_{n})
    \Phi(x_{0},x_{1}\cdots x_{n})(Tf)(x_{n})
    d\mu_{0}(x_{0})d\mu(x_{1})\cdots d\mu(x_{n}) \\
  &=\int_{\Omega}\varphi\cdot (Tf)\circ Z_{n}\text{ } d\mathbb{P}.
\end{align*}
Hence (\ref{E:SP.2}) follows.
\end{proof}

\begin{corollary}
\label{C:SP.0}
Let $(\Omega, \mathcal{F}, \mathbb{P}, (Z_{n}))$ be as in the theorem.  Then 
the process $(Z_{n})$ is Markov.
\end{corollary}
\begin{proof}
We must show that 
\[
  \mathbb{E}(f\circ Z_{n+1} | \mathcal{F}_{n})
  =\mathbb{E}(f\circ Z_{n+1} | Z_{n})
\]
By the theorem, we only need to show that 
\[
  \mathbb{E}(f\circ Z_{n+1} | Z_{n})=(Tf)\circ Z_{n}.
\]
In checking this we use the transition operator $T$.  As a result we may now
assume that $\varphi$ has the form $\varphi=g\circ Z_{n}$ for $g$ a 
measurable function on $S$.  Hence 
\begin{align*}
  \int_{\Omega}\varphi (f\circ Z_{n+1})d\mathbb{P}
  &= \int_{\Omega}(g\circ Z_{n})(f\circ Z_{n+1})d\mathbb{P} 
  = \langle g, Tf\rangle_{L^{2}(\mu)} \\ 
  &= \int_{S}g(Tf)d\mu 
  = \int_{\Omega}(g\circ Z_{n})((Tf)\circ Z_{n})d\mathbb{P} \\
  &= \int_{\Omega}\varphi ((Tf)\circ Z_{n})d\mathbb{P} 
\end{align*}
which is the desired conclusion.
\end{proof}

\begin{definition}
\label{D:SP.2}
We say that a measurable function $f$ on $S$ is harmonic if $Tf=f$.
\end{definition}

\begin{definition}
\label{D:SP.3}
A sequence of random variables $(F_{n})$ is said to be a 
martingale if and only if $\mathbb{E}(F_{n+1} | \mathcal{F}_{n})=F_{n}$ for 
all $n\in \mathbb{N}_{0}$.
\end{definition}

\begin{corollary}
\label{C:SP.1}
Let $(Z_{n})_{n\in \mathbb{N}_{0}}$ be a stochastic process with stationary 
transitions and operator $T$.  Let $f$ be a measurable function on $S$.

Then $f$ is harmonic if and only if $(f \circ Z_{n})_{n\in \mathbb{N}_{0}}$ 
is a martingale.
\end{corollary}
\begin{proof}
This follows from (\ref{E:SP.2}) combined with Definition \ref{D:SP.3}.
\end{proof}

\begin{corollary}
\label{C:SP.2}
Suppose a process $(Z_{n})_{n\in \mathbb{N}_{0}}$ is stationary with a 
fixed transition operator $T:L^{2}(\mu) \to L^{2}(\mu)$.  Then 
$\mu = \mathbb{P} \circ Z_{n}^{-1}$ for all $n\in \mathbb{N}_{0}$.
\end{corollary}
\begin{proof}
Let $f$ and $g$ be a pair of functions on $S$ as specified above.  Then we 
showed that
\[
  \int_{S}gfd\mu = \int_{\Omega}(g\circ Z_{n})(f\circ Z_{n})d\mathbb{P}
\]
which is the desired conclusion.
\end{proof}


\subsection{Martingales and Boundaries}
\label{sec:SP.1a}
Let $G=(G^{0}, G^{1})$ be an infinite graph with a fixed conductance $c$, and 
let the corresponding operators be $\Delta_{c}$ and $T_{c}$.

Let $h:G^{0}\to \mathbb{R}$ is a harmonic function, i.e., $\Delta_{c}h=0$, or 
equivalently $T_{c}h=h$.

As an application of Corollary \ref{C:SP.1}, we may then apply a theorem of 
J. Doob to the associated martingale $h \circ Z_{n}$, $n \in \mathbb{N}_{0}$. 
This means that the sequence $(h \circ Z_{n})$ will then have $\mathbb{P}$- a. 
e. limit i.e.,
\begin{equation}
\label{E:SP.5.1}
  \lim_{n\to \infty}h \circ Z_{n}=v \quad \mbox{pointwise} \quad \mathbb{P} 
  \text{ a.e.}
\end{equation}

The limit function $v:\Omega \to \mathbb{R}$ will satisfy 
$v(x_{0}x_{1}x_{2}\cdots)=v(x_{1}x_{2}x_{3}\cdots)$, or equivalently, 
\begin{equation}
\label{E:SP.5.2}
  v=v\circ \sigma.
\end{equation}

The existence of the limit in (\ref{E:SP.5.1}) holds if one or the other of 
the two conditions is satisfied:
\begin{enumerate}[(i) ]
  \item $h\in L^{\infty}$; or 
  \item $\sup_{n}\int_{\Omega}|h \circ Z_{n}|^{2}d\mathbb{P} < \infty$.
\end{enumerate}

\begin{proposition}
\label{P:SP.1}
\cite{Jor06a} If $h: G^{0} \to \mathbb{R}$ is harmonic and if (i) or (ii) 
hold, then
\begin{equation}
\label{E:SP.5.3}
  h(x)=\int_{\Omega}v\text{ }d\mathbb{P}_{x} \quad \mbox{for all } x\in G^{0},
\end{equation}
where $\mathbb{P}_{x}=$ the measure $\mathbb{P}$ conditioned with 
$Z_{0}(\gamma)=x$.  The converse implication holds as well.
\end{proposition}
\begin{proof}
Starting with $h$ harmonic, if the Doob-limit $v$ in (\ref{E:SP.5.1}) exists, 
then it is clear that $v$ satisfies (\ref{E:SP.5.2}).  By Dominated 
Convergence, (\ref{E:SP.5.3}) will be satisfied.  

Conversely, suppose some measurable $v:\Omega \to \mathbb{R}$ satisfies 
(\ref{E:SP.5.2}), and the integral in (\ref{E:SP.5.3}) exists then
\begin{align*}
  (T_{c}h)(x)&=\sum_{y\sim x}p(x,y)h(y)  \\
  &=\sum_{y\sim x}\mathbb{P}(Z_{0}=x, Z_{1}=y)\mathbb{E}(v|Z_{0}(\cdot)=y) \\
  &\underset{\text{by (\ref{E:SP.5.2})}}{=}  
  \sum_{y\sim x}p(x,y)\mathbb{E}_{x}(v|Z_{1}(\cdot)=y)  \\
  &=\sum_{y\sim x}p(x,y)\mathbb{E}(v|Z_{0}=x, Z_{1}=y) \\
  &=\mathbb{P}_{x}(v(\cdots))  \\
  &=h(x),
\end{align*}
showing that $h$ is harmonic.

\end{proof}

\subsection{Solenoids}
\label{sec:SP.2}

\begin{example}
\label{Ex:SP.0}
Let $S$ be a compact Hausdorff space, and $\sigma:S \to S$ a finite-to-one 
endomorphism onto $S$.  Let $X_{\sigma}(S)$ be the corresponding solenoid:
\[
  X_{\sigma}(S) \subset \prod_{n\in \mathbb{N}_{0}}S, \quad \mbox{where }
  \mathbb{N}_{0}=\{0\}\cup \mathbb{N}=\{0, 1, 2, 3, \cdots\},
\]
\begin{equation}
\label{E:SP.6}
  X_{\sigma}(S)=\{(x_{k})_{k\in \mathbb{N}_{0}}| \sigma (x_{k+1})=x_{k}\}.
\end{equation}

One advantage of a choice of solenoid over the initial endomorphism 
$\sigma:S \to S$ is that $\sigma$ induces an \textit{automorphism} 
$\widehat{\sigma}:X_{\sigma}(S) \to X_{\sigma}(S)$ as follows:
\[
  \widehat{\sigma}((x_{0}x_{1}x_{2}\cdots))
  =(\sigma(x_{0})x_{0}x_{1}x_{2}\cdots),
  \quad \mbox{with inverse} \quad 
  \widehat{\sigma}^{-1}((x_{0}x_{1}x_{2}\cdots))=(x_{1}x_{2}x_{3}\cdots).
\]

Let $W:S \to [0,1]$ be a Borel measurable function, and set 
\begin{equation}
\label{E:SP.7}
  (T_{W}f)(x)=\sum_{\underset{\sigma (y)=x}y}W(y)f(y), \quad f\in B(S), x\in S.
\end{equation}
Assume
\begin{equation}
\label{E:SP.8}
  \sum_{\sigma (y)=x}W(y)\equiv 1, \forall x\in S.
\end{equation}

For points $x\in S$, set $D(x):=\sharp \{y|\sigma(y)=x\}$.  A measure $\mu$ 
on $S$ is said to be \textit{strongly invariant} if 
\[
  \int_{S}\frac{1}{D(x)}\underset{\sigma(y)=x}{\sum_{y}}f(y)d\mu(x)
  =\int_{S}f(x)d\mu(x).
\]

\begin{lemma}
\label{L:SP.1}
Assume a measure $\mu$ on $S$ is strongly invariant, and let $m$ be a function 
on $S$.  Set $Vf(x)=m(x)f(\sigma(x))$.  Then the adjoint operator 
\[
  V^{*}:L^{2}(\mu)\to L^{2}(\mu) \quad \mbox{is} \quad 
  (V^{*}f)(x)=\frac{1}{D(x)}\underset{\sigma(y)=x}{\sum_{y}}\overline{m}(y)f(y)
\]
\end{lemma}
\begin{proof}
See \cite{Jor06a}.
\end{proof}

Set $\Omega:=X_{\sigma}(S)$ and equip it with the $\sigma$-algebra 
$\mathcal{F}$ and the topology which is generated by the cylinder sets.

Set $Z_{k}:\Omega \to S$, 
\begin{equation}
\label{E:SP.9}
  Z_{k}(x_{0}x_{1}x_{2} \cdots):=x_{k}, \quad k \in \mathbb{N}_{0}.
\end{equation}

Let $E \subset S$ be a Borel set, and consider 
\begin{equation}
\label{E:SP.10}
  Z_{k}^{-1}(E)=\{\omega \in \Omega | Z_{k}(\omega) \in E\}.
\end{equation}
Then the $\sigma$-algebra $\mathcal{F}$ on $\Omega$ is generated by the sets 
\begin{equation}
\label{E:SP.11}
  Z_{k}^{-1}(E) \quad \mbox{as $k$ and $E$ vary}.
\end{equation}
Set
\begin{equation}
\label{E:SP.12}
  \mathcal{F}_{n}:=\sigma\text{-algebra}\ll Z_{k}|k \leq n\gg 
\end{equation}
where $\ll \cdot \gg$ refers to the $\sigma$-algebra as specified in 
(\ref{E:SP.10}).

In $\Omega=X_{\sigma}(S)$, consider the following random walk: For points 
$x,y \in S$, a transition $x \to y$ is possible if and only if $\sigma(y)=x$; 
and in this case the transition probability is $p_{W}(x,y):=W(y)$.

Let $\mu$ be a probability measure on $S$.  In $\Omega$ we introduce the 
following Kolmogorov measure $\mathbb{P}:=\mathbb{P}_{W}$ which is 
determined on cylinder sets as follows
\begin{align}
  \mathbb{P}(C_{n})&:=\mathbb{P}(C(E_{0}, E_{1}, E_{2} \cdots, E_{n})) \\
       &=\int_{E_{0}}\int_{E_{1}} \cdots \int_{E_{n}}
       W(x_{1})W(x_{2})\cdots W(x_{n})d\mu(x_{0})d\mu(x_{1})\cdots d\mu(x_{n})
\end{align}
More specifically, $\mathbb{P}$ is a measure on infinite paths, and 
\begin{equation}
\label{E:SP.15}
  C_{n}=\{\omega=(\omega_{0}\omega_{1}\omega_{2}\cdots)|
  \sigma(\omega_{i+1})=\omega_{i}, Z_{k}(\omega) \in E_{k}, \text{ for } 
  0 \leq k \leq n\}. 
\end{equation}

\end{example}

\begin{example}
\label{E:SP.0.1}
The following is a solenoid which is used in both number theory (the study 
of algebraic irrational numbers) and in ergodic systems. \cite{BrJo91}. For
this family of examples, the solenoids are associated with specific 
polynomials $p \in \mathbb{Z}[x]$.

Let $S:=\mathbb{T}^{s}$ where $s \in \mathbb{N}$ is fixed; and let 
$p(x)=a_{0}x^{s}+a_{1}x^{s-1}+\cdots +a_{s}$; $a_{0}\neq 0$, be a polynomial, 
$p \in \mathbb{Z}[x]$.  Set
\[
  F=F_{p}:=
  \begin{pmatrix}
    0 & a_{0} & 0 & 0 & \cdots & 0 \\
    0 & 0 & a_{0} & 0 & \cdots & 0 \\
    0 & 0 & 0 & a_{0} & \cdots & 0 \\
    \vdots & \vdots & \vdots & \vdots & \ddots & \vdots \\
    0 & 0 & 0 & 0 & \cdots & a_{0}  \\
    -a_{s} & -a_{s-1} & \cdots & \cdots & -a_{2} & -a_{1}  
  \end{pmatrix}.
\]

Consider the shift $\sigma$ on the infinite torus 
$\prod_{\mathbb{Z}}\mathbb{T}^{s}=(\mathbb{T}^{s})^{\mathbb{Z}}$, and set 
\[
  X_{\sigma}:=\{(z_{n})_{n \in \mathbb{Z}} \in (\mathbb{T}^{s})^{\mathbb{Z}} | 
  a_{0}z_{n+1}=Fz_{n}\}.  
\]
Then it follows that $X_{\sigma}(p)$ is $\sigma$-invariant and closed.  As a
result, $X_{\sigma}(p)$ is a compact solenoid.

\end{example}

\section{Graphs}
\label{sec:G}
One additional application of these ideas is to infinite graph systems 
$(G, c)$ where $G$ is a graph and $c$ is a positive conductance function.  
A comprehensive study of this class of examples was carried out in the 
paper \cite{JoPe08}.  We will adapt the convention from that paper: 
\begin{itemize}
  \item[$G^{0}$ :] the set of vertices in $G$; 
  \item[$G^{1}$ :] the set of edges in $G$; 
  \item[and]$c$ : $G^{1} \to \mathbb{R}_{+}$ the conductance function.
\end{itemize}

\subsubsection*{Assumptions}
\begin{enumerate}[(i) ]
\item \textit{Edge symmetry.} If $x, y \in G^{0}$ and $(x,y) \in G^{1}$, 
then we assume that $c_{x,y}=c_{y,x}$.  Moreover, $(x,y)\in G^{1} 
\Leftrightarrow (y,x)\in G^{1}$.
\item \textit{Finite neighborhoods.} For all $x \in G^{0}$, the set 
$Nbh(x)=\{y \in G^{0} | (x,y) \in G^{1}\}$ is finite. 
\item \textit{No self-loops.} If $x \in G^{0}$, then $x \notin Nbh(x)$. \\
      Convention: If $x,y \in G^{0}$, we write $x \sim y$ iff 
      $(x,y) \in G^{1}$. 
\item \textit{Connectedness.} For all $x, y \in G^{0}$ there exists 
  $\{ x_{i}\}_{i=0}^{n} \subset G_{0}$ such that $(x_{i},x_{i+1}) \in G^{1}$, 
  $i=0, 1, \cdots, n-1$ $x_{0}=x$ and $x_{n} = y$. 
\item \textit{Choice of origin.} We select an origin $o \in G^{0}$.
\end{enumerate}

\begin{definition}
\begin{itemize}
\item The Laplace operator $\Delta = \Delta_{c}$:
\[
  (\Delta f)(x):=\sum_{y \sim x}c_{x,y}(f(x)-f(y)).
\]
\item Hilbert spaces:
  \begin{enumerate}[(i) ]
    \item $l^{2}(G^{0})$: functions $f: G^{0} \to \mathbb{C}$ such that
    $\|f\|_{2}^{2} = \sum_{x \in G^{0}} |f(x)|^{2} < \infty$. 
    Set $\langle f_{1}, f_{2}\rangle_{2}:=\sum_{x \in G^{0}}\overline{f_{1}(x)}f_{2}(x)$. For every $x \in G^{0}$, set $\delta_{x}:G^{0} \to \mathbb{R}$,
    \[
      \delta_{x}(y)=
      \begin{cases}
        1 & \text{if $y=x$} \\
        0 & \text{if $y \neq x$}
      \end{cases}
    \]
    Note that $\{\delta_{x}\}$ is an orthonormal basis (ONB) in 
    $l^{2}(G^{0})$.
    \item $\mathcal{H}_{E}$: finite energy functions module constants: 
    \begin{equation}
    \label{E:1.1} 
      \|f\|_{E}^{2} = \frac{1}{2}\sum_{\text{all}}
      \sum_{x \sim y}c_{x,y}|f(x)-f(y)|^{2}.
    \end{equation}
    Set
    \begin{equation}
    \label{E:1.2} 
    \langle f_{1}, f_{2}\rangle_{E} := \frac{1}{2}\sum \sum_{x \sim y}c_{x,y}
    (\overline{f_{1}(x)}-\overline{f_{1}(y)})(f_{2}(x)-f_{2}(y)).
    \end{equation}
  \end{enumerate}  
\item Dipoles. 
  For all $x \in G^{0}$ there is a unique 
  $v_{x} \in \mathcal{H}_{E}$ such that 
  \[
    \langle v_{x},f\rangle_{E}=f(x)-f(o), \quad \forall f\in \mathcal{H}_{E}.
  \]
  In this case, $v_{x}$ satisfies
  $\Delta v_{x} = \delta_{x}-\delta_{o}$,
  and we make the choice $v_{x}(o)=0$.  The function 
  $v_{x}: G^{0} \to \mathbb{R}$ is called a \textit{dipole}.
\end{itemize}
\end{definition}

\begin{example}
\textit{The dyadic tree.}
\begin{itemize}
\item $\mathcal{A}=$ the alphabet of two letters, bits 
      $\{0,1\}\simeq \mathbb{Z}_{2}$. 
\item $G^{0}$: the set of all finite words in $\mathcal{A}: o=\emptyset=$ 
the empty word, 
$x=(a_{1}a_{2} \cdots a_{n})\in G^{0}$, $a_{i} \in \mathcal{A}$, a word of 
length $n$; $l(x)=n$.
\item $G^{1}:=$ the edges in the dyadic tree.  If $x = \emptyset$, 
$Nbh(x)=\{0,1\}$ two one-letter words.  If $l(x)=n > 0$, 
$x=(a_{1}a_{2} \cdots a_{n})$, 
$Nbh(x)=\{(a_{1} \cdots a_{n-1}), (x0), (x1)\}$. Set 
$x^{*}:=(a_{1} \cdots a_{n-1})$. 
\item \textit{Constant conductance.} \\
This is the restriction $c \equiv 1$ on $G^{1}$.  Then
\[
  (\Delta f)(o)=2f(o)-f(0)-f(1),  \quad \mbox{and}
\]
\[
  (\Delta f)(x)=3f(x)-f(x^{*})-f(x0)-f(x1), 
\]
if $x \in G^{0}$, and $l(x) >0$. 
\item \textit{Paths in the tree.} 
If $x=(a_{1}a_{2} \cdots a_{n}) \in G^{0}$, there is a unique path $\gamma(x)$ 
from $\emptyset$ to $x$: the path is 
\[
  \gamma(x)=\{(o,a_{1}), (a_{1}, (a_{1}a_{2})), \cdots ((a_{1} \cdots a_{n-1}), x)\}
\] and consists of $n$ edges.
\item \textit{Concatenation of words:} For $x=(a_{1}a_{2} \cdots a_{n})$, 
$y=(b_{1}b_{2} \cdots b_{m}) \in G^{0}$.  Set 
$z=z(xy)=(a_{1} \cdots a_{n} b_{1} \cdots b_{m})$.
\end{itemize}
\end{example}

The dipoles $(v_{x})$ are indexed by $x \in G^{0}\setminus (o)$, and 
$v_{x}(o)=0$ where $o$ is the chosen origin.  If $G=$ the tree, then 
$o = \emptyset =$ the empty word.

\begin{lemma}
\label{L:1.3}
\cite{JoPe08} Let $x=(a_{1}a_{2} \cdots a_{n})$, $a_{i} \in A$, $n=l(x)$;
and $y=(b_{1}b_{2} \cdots b_{m})$, $b_{i} \in A$, $m=l(y)$.  Then
\begin{enumerate}[(i) ]
  \item 
  \[
    v_{x}(y):=
    \begin{cases}
      0 & \text{if $y=o$} \\
      2^{n-m}\cdot(2^{m}-1)-\frac{2^{n}-1}{2} & \text{if $m \leq n$} \\
      \frac{2^{n}-1}{2} & \text{if $m > n$}
    \end{cases}
  \] 
  \item $v_{x} \in \mathcal{H}_{E}$, and 
  $\|v_{x}\|_{E}^{2}=\frac{2}{3}(2^{2n}-1)$. \\
  \item $\langle v_{x}, v_{y}\rangle_{E} 
  = \frac{2}{3}(2^{2\min(l(x)l(y))}-1)=\#(\gamma(x) \cap \gamma(y))$, 
  for all $x, y \in G^{0}\setminus (o)$.
\end{enumerate}

\end{lemma}
\begin{proof}
\begin{enumerate}[(i) ]
  \item By the uniqueness in Lemma \ref{L:1.3}, it is enough to prove that the 
  function $v_{x}$ in (i) satisfies
  $\langle v_{x},f\rangle_{E}=f(x)-f(o)$  for all $f\in \mathcal{H}_{E}$, 
  and therefore also
  \begin{equation}
  \label{eq:1.1}
    \Delta v_{x}=\delta_{x} - \delta_{o};
  \end{equation}
  and that (ii)-(iii) hold.

  Specifically, we must prove that 
  \begin{align*}
    (\Delta v_{x})(o) &= -1,  \\
    (\Delta v_{x})(x) &= 1, \text{ and }  \\
    (\Delta v_{x})(y) &= 0, \text{ if $y \notin \{o, x \}$}.
  \end{align*}  
  Each is a computation:
  \begin{align*}
    (\Delta v_{x})(o) &= 2v_{x}(o)-v_{x}(0)-v_{x}(1) \\
      &= 0-(2\cdot2^{n-1}-(2^{n}-1)) \\
      &= 1  \\
      &= \delta_{o}(o). 
  \end{align*}
  And if $y \neq o$, but $m < n$, then
  \begin{align*}
    (\Delta v_{x})(y) &= 3v_{x}(y)-v_{x}(y^{*})-v_{x}(y0)-v_{x}(y1) \\
      &= 3\cdot2^{n-m}\cdot(2^{m}-1) - 2^{n-m+1}\cdot(2^{m-1}-1)
         - 2\cdot2^{n-m-1}\cdot(2^{m+1}-1) \\
      &= 0.
  \end{align*}
  
  Finally, we compute the case $y=x$ as follows:
  \begin{align*}
    (\Delta v_{x})(x) &= 3v_{x}(x)-v_{x}(x^{*})-v_{x}(x0)-v_{x}(x1) \\
      &= 3\cdot(2^{n}-1) - 2\cdot(2^{n-1}-1) - 2\cdot(2^{n}-1) \\
      &= 0-3+2+2 = 1 \\
      &= \delta_{x}(x) - \delta_{o}(x).
  \end{align*}
  We leave the case $m=l(y)>n$ to the reader. 
  \item  A computation using (\ref{E:1.1}) yields
  \begin{align*}
    \|v_{x}\|^{2}_{E} &= \frac{1}{2}\sum_{m \leq n} (2^{n-m})^{2} \\
      &= \frac{1}{2} \cdot 2^{2n} \cdot \left( \frac{1-2^{-2n}}
      {1-2^{-2}}\right) \\
      &= \frac{2}{3}(2^{2n}-1)
  \end{align*}
  proving (ii). 
  \item Suppose $m=l(y)<n=l(x)$, $x,y \in G^{0} \setminus (o)$.  
  From (\ref{E:1.2}), we see that the contribution to 
  $\langle v_{x}, v_{y}\rangle_{E}$ only includes words $z$ with $l(z)\leq m$. 
\end{enumerate}

The desired conclusion 
\[
  \langle v_{x}, v_{y}\rangle_{E} = 2^{-2m} \#(\gamma(x) \cap \gamma(y))
\] 
follows as in (ii).  The possibilities may be illustrated in Figure 1 below.


\end{proof}

\section{Specific Transition Operators}
\label{sec:STO}
\subsection{Transition on Graphs}
\label{sec:STO.1}
Let $G=(G^{0}, G^{1})$ be a graph with conductance function 
$c:G^{1} \to \mathbb{R}^{+}$, and transition probabilities
\[
  p(x,y):=\frac{c(x,y)}{c(x)}, \quad \forall (x,y)\in G^{1}.
\]

Note that $c(x)p(x,y)=p(x,y)c(y)$, which makes the corresponding $p$-random 
walk reversible.

\begin{lemma}
\label{L:STO.1}
Assume that $\# Nbh(x)<\infty$ for all $x\in G^{0}$.  Set 
\[
  (Tf)(x):=\sum_{y \sim x}p(x,y)f(y), 
\]
and let $(Z_{n})$ be the random walk on 
$G^{0}$ with transition probabilities $p(x,y)$ on edges $(xy)$ in $G$, i.e.,
\[
  \mathbb{P}(\{\gamma |Z_{n}(\gamma)=x, Z_{n+1}(\gamma)=y\})=p(x,y) \quad
  \mbox{for $(xy) \in G^{1}$} 
\]

Let $T$ be the transition operator, and for $\varphi \in l^{1}(G^{0})$, set 
\[
  \langle \varphi \rangle := \sum_{x \in G^{0}}\varphi(x)
\]
then for pairs of functions $f_{1}$ and $f_{2}$ on $G^{0}$, we have
\[
  \mathbb{E}((f_{1}\circ Z_{n})\cdot (f_{2}\circ Z_{n+1}))
  =\langle T^{n}(f_{1}\cdot Tf_{2}) \rangle
\]
with $f_{1}$ and $f_{2}$ are restricted to make the last sum convergent.
\end{lemma}
\begin{proof}
Let $f_{1}$, $f_{2}$ be a pair of functions (real valued) on $G^{0}$ such 
that the pointwise product $f_{1}\cdot (Tf_{2})$ is in $l^{1}(G^{0})$.  Then 
for $n\in \mathbb{N}_{0}$, we now compute the $Z_{n}$-expectations: For the 
$\mathbb{P}$-integration on path space $\Omega$, we have:
\begin{align*}
  &\mathbb{E}((f_{1}\circ Z_{n})\cdot (f_{2}\circ Z_{n+1}))\\
  &=\int_{\Omega}(f_{1}\circ Z_{n})\cdot (f_{2}\circ Z_{n+1})d\mathbb{P} \\
  &=\underset{\text{such that } x_{i-1} \sim x_{i}} 
    {\sum_{x_{0}}\sum_{x_{1}}\cdots \sum_{x_{n+1}}} 
    p(x_{0},x_{1})p(x_{1},x_{2})\cdots p(x_{n},x_{n+1})f_{1}(x_{n})
    f_{2}(x_{n+1}) \\
  &=\sum_{x_{0}}\sum_{x_{1}}\cdots \sum_{x_{n}} 
    p(x_{0},x_{1})p(x_{1},x_{2})\cdots p(x_{n-1},x_{n})f_{1}(x_{n})
    (Tf_{2})(x_{n}) \\
  &=\sum_{x_{0}\in G^{0}} T^{n}(f_{1}\cdot Tf_{2})(x_{0}) \\ 
  &=\langle T^{n}(f_{1}\cdot Tf_{2}) \rangle
\end{align*}
\end{proof}

\begin{theorem}
\label{T:STO.1}
Let $(G,c)$ be a graph with conductance $c:G^{1}\to \mathbb{R}_{+}$.  Assume 
that $\sharp Nbh(x) < \infty$ for all $x\in G^{0}$, when 
$Nbh(x):=\{y\in G^{0}|y \sim x\}$.  Set 
\[
  p(x,y):=\frac{c(x,y)}{c(x)} \quad \text{and} \quad 
  (Tf)(x):=\sum_{y\sim x}p(x,y)f(y).
\]  
Set 
\[
  l^{1}(G^{0},\mu_{c})=\{f:G^{0}\to \mathbb{R}| x\to c(x)f(x)\in l^{1}(G^{0})\},
  \quad \mbox{and} \quad \langle f\rangle_{c}:=\sum_{x\in G^{0}}c(x)f(x).
\]
Let $\mathbb{P}^{(c)}=\mathbb{P}^{(\mu_{c})}$ be the cylinder path-measure 
on
\[
  \Omega:=\{(x_{0}x_{1}x_{2}\cdots)|x_{i}\in G^{0}, x_{i-1}\sim x_{i}, 
  i \in \mathbb{N}\}
\]
where we use $\mu_{c}$ in the first variable $x_{0}$, and counting measure 
on the remaining variables.  Then
\[
  \mathbb{E}^{(\mu_{c})}((f_{1}\circ Z_{n})\cdot (f_{2}\circ Z_{n+1})) 
  =\langle f_{1}\cdot Tf_{2} \rangle_{c}
\]
\end{theorem}
\begin{proof}
\begin{align*}
  &\mathbb{E}^{(\mu_{c})}((f_{1}\circ Z_{n})\cdot (f_{2}\circ Z_{n+1})) \\
  &=\underset{\text{such that } x_{i-1} \sim x_{i}} 
    {\sum_{x_{0}}\sum_{x_{1}}\cdots \sum_{x_{n+1}}}c(x_{0}) 
    p(x_{0},x_{1})p(x_{1},x_{2})\cdots p(x_{n},x_{n+1})f_{1}(x_{n})
    f_{2}(x_{n+1}) \\
  &=\sum_{x_{0}}\sum_{x_{1}}\cdots \sum_{x_{n}}c(x_{0}) 
    p(x_{0},x_{1})p(x_{1},x_{2})\cdots p(x_{n-1},x_{n})f_{1}(x_{n})
    (Tf_{2})(x_{n}) \\
  &=\sum_{x_{0}}c(x_{0})T^{n}(f_{1}\cdot Tf_{2})(x_{0}) \\ 
  &=\langle T^{n}(f_{1}\cdot Tf_{2}) \rangle_{c} \\
  &=\langle f_{1}\cdot Tf_{2} \rangle_{c}
\end{align*}

In the multiple summations $\sum_{x_{0}}\sum_{x_{1}}\cdot \sum_{x_{n+1}}$, it 
is just the first $\sum_{x_{0}}$-summation that is possibly infinite; in case 
the vertex-set $G^{0}$ is infinite.  Note that the combined summations in the 
beginning of the proof contribute the integration over the set $\Omega$ of 
all infinite paths $\gamma=(x_{0}x_{1}x_{2}\cdots)$ specified by 
$x_{0} \sim x_{1}$, $x_{1} \sim x_{2}$, $x_{2} \sim x_{3}$, $\cdots,$ at 
each step, moving from $x_{i}$ to the next variable, note that $x_{i+1}$ 
ranges over the finite set $Nbh(x_{i})$.  For more details on this point, see 
(\ref{E:STO.1}), below.

In the last step, we used the following formula which is valid on 
$l^{1}(\mu_{c})$:
\begin{equation}
\label{E:STO.1}
  \langle T\varphi \rangle_{c}=\langle \varphi \rangle_{c}, \quad
  \varphi \in l^{1}(\mu_{c}). 
\end{equation}
We prove (\ref{E:STO.1}):
\begin{align*}
  \langle T\varphi \rangle_{c}
  &=\sum_{x\in G^{0}}c(x)\sum_{y\sim x}p(x,y)\varphi(y) \\
  &=\sum_{y\in G^{0}}\varphi(y)\sum_{x\sim y}c(x,y) \\
  &=\sum_{y\in G^{0}}\varphi(y)c(y) \\
  &=\langle \varphi \rangle_{c}.
\end{align*}
\end{proof}

\subsection{Transfer Operators}
\label{sec:STO.2}
In section \ref{sec:SP}, we showed that a stochastic process 
$(Z_{n})_{n\in \mathbb{N}_{0}}$ on a probability space 
$(\Omega, \mathcal{F}, \mathbb{P})$ induces a transfer operator $T$.  The 
derivation of $T$ is then essentially canonical.

Here, the strategy will be reversed; but now, starting with $T$, there is a 
variety of choices of associated processes $(Z_{n})_{n\in \mathbb{N}_{0}}$.

\subsubsection{Setting}
\label{sec:STO.2.1}
Let $S$ be a compact Hausdorff space.  Let $(S, \mathcal{B})_{S}, \mu)$ be 
a Borel probability measure space, and let $p:S\times S\to \mathbb{R}_{\geq 0}$
be a continuous function such that
\begin{equation}
\label{E:STO.2.1}
  \int_{S}p(x,y)d\mu(y)\equiv 1 \quad \mbox{$\mu$ a.e. $x$}.
\end{equation}
Set
\begin{equation}
\label{E:STO.2.2}
  (Tf)(x):=\int_{S}p(x,y)d\mu(y) \quad \mbox{for all $f\in L^{\infty}(S)$}.
\end{equation}
Set
\begin{equation}
\label{E:STO.2.3}
  \Omega:=\Omega_{p}=\{\gamma=(x_{0}x_{1}x_{2}\cdots )| x_{i}\in S, 
  \text{ s.t. } p(x_{i-1}, x_{i})>0\},
\end{equation}
so an infinite path-space with path transitions governed by te function $p$.

Let $\mathbb{P}=(\mathbb{P}_{p})$ be the associated cylinder measure on 
$\Omega_{p}$ as defined in section \ref{sec:SP}.  For $n\in \mathbb{N}_{0}$ 
and $\gamma=(x_{0}x_{1}x_{2}\cdots ) \in \Omega_{p}$, set
\begin{equation}
\label{E:STO.2.4}
  Z_{n}(\gamma):=x_{n}; \quad \mbox{i.e., $Z_{n}:\Omega_{p}\to S$}
\end{equation}
is an $S$ valued random variable for all $n \in \mathbb{N}_{0}$.

\begin{theorem}
\label{T:STO.1}
Let $p:S\times S \to \mathbb{R}_{\geq 0}$ be as stated in (\ref{E:STO.2.1}) 
above.  Let $T$ be the transfer operator (\ref{E:STO.2.2}).  Then the 
stochastic process $(Z_{n})_{n\in \mathbb{N}_{0}}$ in (\ref{E:STO.2.4}) 
satisfies
\begin{equation}
\label{E:STO.2.5}
  \mathbb{E}^{(p)}((f_{1}\circ Z_{n})\cdot (f_{2}\circ Z_{n+1})) 
  =\int_{S}(T^{n}(f_{1}\cdot Tf_{2}))(x)d\mu (x) \quad 
  \mbox{for all $f_{1}, f_{2} \in L^{\infty}(S)$}.
\end{equation}
\end{theorem}
\begin{proof}
The details in the computation for (\ref{E:STO.2.5}) follow those in section
\ref{sec:SP}, but the reasoning is now reversed.  Indeed,
\begin{align*}
  &\mathbb{E}^{(p)}((f_{1}\circ Z_{n})\cdot (f_{2}\circ Z_{n+1}))  \\
  &=\int_{S}\int_{S}\cdots \int_{S}p(x_{0},x_{1})\cdots p(x_{n},x_{n+1})
    f_{1}(x_{n})f_{2}(x_{n+1})d\mu(x_{0})d\mu(x_{1})\cdots d\mu(x_{n+1}) \\
  &=\int_{S}\int_{S}\cdots \int_{S}p(x_{0},x_{1})\cdots p(x_{n-1},x_{n})
    f_{1}(x_{n})(Tf_{2})(x_{n})d\mu_{0}(x_{0})d\mu(x_{1})\cdots d\mu(x_{n}) \\
  &=\int_{S}(T^{n}(f_{1}\cdot Tf_{2}))(x)d\mu (x).  
\end{align*}
\end{proof}

\begin{definition}
\label{D:STO.1}
Let $T$ be a transition operator satisfying the conditions (\ref{E:STO.2.1}) 
and (\ref{E:STO.2.2}), and suppose there is a Perron-Frobenius measure 
$\mu_{0}$ on $S$, i.e.,
\begin{equation}
\label{E:STO.2.6}
  \mu_{0} \circ T=\mu_{0}.
\end{equation}
We say that $T$ is \textit{ergodic} if there is only one probability measure 
$\mu_{0}$ on $(S, \mathcal{B}_{S})$ which solves (\ref{E:STO.2.6}).  

If $T$ is ergodic, and $\mu_{0}$ is the (unique) Perron-Frobenius measure, 
then it follows from the Pointwise Ergodic Theorem that for all
$f \in L^{\infty}(S)$, the limit
\begin{equation}
\label{E:STO.2.7}
  \lim_{n\to \infty}T^{n}(f)=\mu_{0}(f)\Iidentity
\end{equation}
pointwise a.e. exits on $S$, where $\Iidentity$ denotes the constant 
function $1$ on $S$.
\end{definition}

\begin{corollary}
\label{C:STO.1}
Let $p$, $T$, $S$, $\mathcal{B}_{S}$, $\mu$, and $(Z_{n})$ satisfy the 
conditions of the theorem.  Further assume $T$ is ergodic with 
Perron-Frobenius measure $\mu_{0}$.  Then 
\begin{equation}
\label{E:STO.2.8}
  \lim_{n\to \infty} \mathbb{E}^{(p)}((f_{1}\circ Z_{n})\cdot 
  (f_{2}\circ Z_{n+1})) = \mu_{0}(f_{1}\cdot Tf_{2})
\end{equation}
is satisfied for all $f_{1}, f_{2} \in L^{\infty}(S)$.
\end{corollary}
\begin{proof}
To verify (\ref{E:STO.2.8}), note that 
$\mathbb{E}^{(p)}((f_{1}\circ Z_{n})\cdot (f_{2}\circ Z_{n+1}))$ is already  
computed in (\ref{E:STO.2.5}) in the theorem.

Since $\mu$ is a probability measure, the conclusion (\ref{E:STO.2.8}) now 
follows from (\ref{E:STO.2.7}), i.e., form an application of the Ergodic 
Theorem.
\end{proof}

\subsection{Transition on Solenoids}
\label{sec:STO.3}
Let $(S, \mu)$ be a measure space, $\sigma : S \to S$ an endomorphism as 
specified in section \ref{sec:SP}.  Let $\Omega:=X_{\sigma}(S)$ be the 
corresponding solenoid.  Let $W:S \to [0,1]$ be a function satisfying
\begin{equation}
\label{E:STO.3}
  \sum_{y, \sigma(y)=x}W(y)=1;
\end{equation}
and let $\mathbb{P}=\mathbb{P}_{\mu, \sigma, W}$ be the corresponding path 
measure.

\begin{lemma}
\label{L:STO.2}
For the solenoid set $Z_{n}:\Omega \to S$, 
$Z_{n}(x_{0},x_{1},x_{2}, \cdots)=x_{n}$, and 
$(Tf)(x)=\sum_{y, \sigma(y)=x}W(y)f(y)$, for $x \in S$.  Suppose $T$
has a Perron-Frobenius measure $\mu_{0}$.  Then 
$(Z_{n})_{n\in \mathbb{N}_{0}}$ is stationary with transition operator $T$.
\end{lemma}
\begin{proof}
Let $f_{1}$, $f_{2}$ be a pair of functions on $S$ satisfying the conditions
listed above.  For the $\mathbb{P}$-integration on path space 
$\Omega(=X_{\sigma}(S))$ we then have:
\begin{align*}
  &\mathbb{E}^{(\mu_{0})}((f_{1}\circ Z_{n})\cdot (f_{2}\circ Z_{n+1})) \\
  &=\int_{S}\underset{\sigma(x_{1})=x_{0}}{\sum_{x_{1}}} 
  \underset{\sigma(x_{2})=x_{1}}{\sum_{x_{2}}} \cdots 
  \underset{\sigma(x_{n+1})=x_{n}}{\sum_{x_{n+1}}}
  W(x_{1})W(x_{2})\cdots W(x_{n+1})f_{1}(x_{n})f_{2}(x_{n+1})d\mu_{0}(x_{0}) \\
  &=\int_{S}\underset{\sigma(x_{1})=x_{0}}{\sum_{x_{1}}} 
  \underset{\sigma(x_{2})=x_{1}}{\sum_{x_{2}}} \cdots 
  \underset{\sigma(x_{n})=x_{n-1}}{\sum_{x_{n}}}
  W(x_{1})W(x_{2})\cdots W(x_{n})f_{1}(x_{n})(Tf_{2})(x_{n})d\mu_{0}(x_{0}) \\
  &=\int_{S}(T^{n}(f_{1}\cdot Tf_{2}))(x_{0})d\mu_{0}(x_{0}) \\
  &=\mu_{0}(T^{n}(f_{1}\cdot Tf_{2}))=\mu_{0}(f_{1}\cdot Tf_{2}) \\
  &=\langle f_{1}, Tf_{2}\rangle_{L^{2}(\mu_{0})}.
\end{align*}

\end{proof}

\subsection{Encodings}
\label{sec:STO.4}
Let $G=(G^{0}, G^{1})$ be a graph where we write $G^{0}$ for the vertices and  
$G^{1}$ for the edges.  Let $S$ be a set.  We say that $G$ yields an encoding 
of the points in $S$ if there are mappings 
\begin{align}
  \tau^{0}&: G^{0} \to S, \quad \mbox{ onto, and } \\
  \tau^{1}&: G^{0} \to \text{Functions }(S \to S)
\end{align}
such that for every $e=(x,y)\in G^{1}$ we have 
\begin{equation}
\label{E:STO.4}
  \tau^{0}(y)=\tau^{1}(e)\tau^{0}(x).
\end{equation}

\subsubsection*{Examples}
\label{sec:STO.4.3.1}
$G=$ the binary tree, 
\begin{align}
  S&=\mathbb{N}_{0}=\{0, 1, 2, \cdots \} \\
   &=\{\sum_{k=0}^{\text{Finite}}x_{k}2^{k}| x_{k}\in \{0,1\}\}
\end{align}
If $n\in \mathbb{N}_{0}$ is given the finite word $(x_{0}x_{1}x_{2}\cdots )$ in
(4.6) is computed from the Euclidean algorithm for division with 2.  

Points in $G^{0}$ are represented by the empty word $o$, and by all finite 
words $w=(x_{0}x_{1}\cdots x_{p})$.  Set 
\begin{equation}
\label{E:STO.7}
  \tau^{0}(w)=\sum_{k=0}^{p}x_{k}2^{k}=n \in \mathbb{N}_{0}.
\end{equation}
Starting with $w=(x_{0}x_{1}\cdots x_{p}) \in G^{0}$, the three neighbors are 
$(w0), (w1)$, and $w^{*}:=(x_{0}x_{1}\cdots x_{p-1})$ truncation, see Figure 2.



Set 
\begin{equation}
\label{E:STO.8}
  \begin{cases}
    \tau^{1}(e_{0}):=n \mapsto n ;  &\text{see } (\ref{E:STO.7}); \\
    \tau^{1}(e_{1}):=n \mapsto n+2^{p+1} ;  &\text{and } \\
    \tau^{1}(e^{*}):=n \mapsto \sum_{k=0}^{p-1}x_{k}2^{k}.   
  \end{cases}
\end{equation}
Note that in this example, there is an additional pair of mappaings 
$\mathbb{N}_{0} \to \mathbb{N}_{0}$

\begin{equation}
\label{E:STO.9}
  \begin{cases}
    \sigma^{0}(n)=2n \\
    \sigma^{1}(n)=2n+1 
  \end{cases}
\end{equation}
corresponding to the encoding mappings:
\begin{equation}
\label{E:STO.10}
  \begin{cases}
    \sigma_{0}:(x_{0}x_{1}\cdots x_{p}) \mapsto 
      \underbrace{(0x_{0}x_{1}\cdots x_{p})}_{\text{one step longer}} \\
    \sigma_{1}:(x_{0}x_{1}\cdots x_{p}) \mapsto (1x_{0}x_{1}\cdots x_{p})
  \end{cases}
\end{equation}

\begin{remark}
\label{R:STO.1}
The same construction works mutatis mutandis with $N$'adic scaling rather 
than the dyadic representation of points in $\mathbb{N}_{0}$.  Moreover, in 
the representation
\begin{equation}
\label{E:STO.11}
  n=\sum_{k=0}^{p}x_{k}N^{k},
\end{equation}
the choices for $x_{k}$ may be from any complete set of residues modulo $N$, 
i.e., points in $\mathbb{N}_{0}/N\cdot \mathbb{N}_{0}$, or 
$\mathbb{Z}/N\mathbb{Z}=$ the cyclic group of order $N$. The residues 
$\{0, 1, \cdots, N-1\}$ is only one choice of many.
\end{remark}

\subsubsection*{Encoding of $\mathbb{Z}$}
\label{sec:STO.4.2}
The representation used in (\ref{E:STO.7}) above works for $\mathbb{Z}$ as 
well, but with the following modification:
\begin{equation}
\label{E:STO.12}
  \tau^{0}(x_{0}x_{1}x_{2}\cdots x_{p}):=-2^{p}+\sum_{k=0}^{p}x_{k}2^{k}.
\end{equation}
Explanation:
\begin{align*}
  \tau^{0}(\underbrace{111 \cdots 1}_{p+1 \text{ times}})
  &=-2^{p}+\sum_{k=0}^{p}x_{k}2^{k} \quad 
  \mbox{with $x_{k}=1$, $0\leq k \leq p$}\\
  &=-2^{p}+2^{p+1}-1 \\
  &=2^{p}-1.
\end{align*}
Hence, with this convention we arrive at an encoding of $\mathbb{Z}$.

\subsubsection*{Graphs vs Compactification:}
\label{sec:STO.4.2a}
In the examples, we represent points in the vertex sets $G^{0}$ on a graph 
$G$ by finite words in a specific finite alphabets.  A choice of 
compactification $\Omega$ of $G^{0}$ is the set of infinite paths $\gamma$, 
i.e., $\gamma=(x_{0}x_{1}x_{2} \cdots)$ where $x_{i}\in G^{0}$, and 
$(x_{i-1}, x_{i})\in G^{1}$ for all $i \in \mathbb{N}$.

In each of the examples we present, we build measure $\mathbb{P}$ on the 
compactifications $\Omega$ with use of Kolmogorov's extension principle. 
This is a projective limit construction which proceeds in three steps
\cite{Jor06a}:
\begin{enumerate}[(i) ]
  \item First specify $\mathbb{P}$ only on finite words, i.e., on cylinder 
    sets over $G^{0}$ 
  \item Check that the prescription of $\mathbb{P}$ on cylinders is 
    consistent. 
  \item With Kolmogorov's theorem than extend $\mathbb{P}$ to the Borel 
    $\sigma$-algebra of subsets in $\Omega$ generated by the cylinder-sets
    \cite{Kol77, Jor06a}.
\end{enumerate}

\begin{definition}
\label{D:STO.2}
In later applications, the following two cases for $\mathbb{P}$ will play a 
role: Consider the subset $\Omega_{\text{Fin}}$ in $\Omega$ consisting of 
paths $\gamma=(x_{0}x_{1}x_{2} \cdots)$ which terminate in infinite 
repetitions, i.e., $\gamma \in \Omega_{\text{Fin}} \Leftrightarrow \exists$ 
$n$ such that $x_{i}=x_{n}$ $\forall$ $i>n$.  The measure $\mathbb{P}$ is 
said to be \textit{tight} if and only if $\mathbb{P}(\Omega_{\text{Fin}})=1$.  
Alternatively, $\mathbb{P}(\Omega_{\text{Fin}})<1$. 
\end{definition}

\subsubsection*{Examples Resumed: }
\label{sec:STO.4.3}
Wavelets. We adopt the standard terminology for dyadic wavelets in 
$L^{2}(\mathbb{R})$, specifically $\varphi$ for a choice of scaling function; 
see \cite{Jor06a}.  Let $(a_{k})_{k\in \mathbb{Z}}$ represent a wavelet 
filter, i.e., satisfying the following three conditions:
\begin{equation}
\label{E:STO.13}
  \sum_{k\in\mathbb{Z}}\overline{a}_{k}a_{k+2l}=\frac{1}{2}\delta_{0,l},
\end{equation}
\begin{equation}
\label{E:STO.14}
  \sum_{k\in\mathbb{Z}}a_{k}=1, \quad \mbox{and}
\end{equation}
\begin{equation}
\label{E:STO.15}
  \varphi(x)= 2\sum_{k\in\mathbb{Z}}a_{k}\varphi(2x-k).
\end{equation}
The function $\varphi$ is in $L^{2}(\mathbb{R})$ and 
\begin{equation}
\label{E:STO.16}
  \int_{\mathbb{R}}\varphi(x)dx=1
\end{equation}
is a chosen normalization.

Let $\widehat{\varphi}$ be the $\mathbb{R}-$ Fourier transform.

The following result is from \cite{Jor06a}.  Let $\Omega:=$ the set of all 
infinite words, and view $\Omega$ as a compactification of the vertex set 
$G^{0}$ of all finite dyadic words.

\begin{lemma}
\label{L:STO.3}
For every $t \in \mathbb{R}$, there is a measure $\mathbb{P}_{t}$ on $\Omega$ 
such that
\begin{equation}
\label{E:STO.17}
  \mathbb{P}_{t}(x_{0}x_{1}\cdots x_{p})= 
  \left\vert\hat{\varphi}(t+\tau^{0}(x_{0}x_{1}\cdots x_{p}))\right\vert^{2}
\end{equation}
where $\tau^{0}: G^{0} \to \mathbb{Z}$ is the encoding of (\ref{E:STO.12}).
\end{lemma}

\begin{lemma}
\label{L:STO.4}
(See \cite{Jor06a}.) 
  \begin{enumerate}[(a) ]
    \item Consider the process $(Z_{n})$ in $(\Omega, \mathbb{P}_{t})$ from 
      (\ref{E:STO.17}) with 
      \[
        Z_{n}\underbrace{(x_{0}x_{1}x_{2}\cdots)}_{\text{infinite word}}:=
        x_{n} \in \{0,1\}.
      \]  
      Then there is a transfer operator $T$ such that 
      the process is $T$-stationary. 
    \item Let 
      \begin{equation}
      \label{E:STO.18} 
        W(e^{it}):=\widetilde{W}(t)=
        \left\vert \sum_{k \in \mathbb{Z}}a_{k}e^{ikt}\right\vert^{2},
      \end{equation}
      where functions $W$ on $\mathbb{T}$ are identified with $2\pi$-periodic 
      functions $\widetilde{W}$ on $\mathbb{R}$, and where $(a_{k})$ is some 
      wavelet filter as in (\ref{E:STO.13})-(\ref{E:STO.15}).
      The transfer operator $T$ is then given by
      \[
        (T_{W}f)(t)=W(\frac{t}{2})f(\frac{t}{2})
        +W(\frac{t}{2}+\pi)f(\frac{t}{2}+\pi).
      \]
      We say that $W$ has scaling-degree $2$. 
      
      Following (\ref{E:STO.9}), let a transition from $n$ to $n+1$ be given 
      by a choice of $x \in \{0,1\}$.  

      Then
      \begin{equation}
      \label{E:STO.19} 
        \mathbb{E}_{t}(Z_{n}Z_{n+1})=\widetilde{W}(t+x\pi).
      \end{equation}
  \end{enumerate}
\end{lemma}

\begin{proposition}
\label{P:STO.1}
Let $\varphi \in L^{2}(\mathbb{R})$ satisfying (\ref{E:STO.15}), and suppose 
$\|\varphi\|_{2}\leq 1$. Let $\Omega$ be the compactification derived from 
the encoding $\tau^{0}$ of $\mathbb{Z}$ in (\ref{E:STO.12}) and let 
$t \in (-\pi, \pi]$.  Let $\mathbb{P}_{t}$ be the measure on $\Omega$ from 
(\ref{E:STO.17}).

\subsubsection*{Part I}
Then the following affirmations are equivalent:
\begin{enumerate}[(a) ]
  \item The translates $\{\varphi(\cdot -k)|k\in \mathbb{Z}\}$ form an 
    orthonormal family in $L^{2}(\mathbb{R})$.
  \item The measures $\mathbb{P}_{t}$ are tight measures on $\Omega$ for all 
    $t$.
  \item $\sum_{n \in \mathbb{Z}}\left\vert 
    \widehat{\varphi}(t+n)\right\vert^{2}= 1$ for all $t\in \mathbb{R}$.  

\end{enumerate}
\subsubsection*{Part II}
If the measures $\mathbb{P}_{t}$ are not tight, then the translates 
$\{\varphi(\cdot -k)\}_{k\in \mathbb{Z}}$ still form a Parseval frame for 
the closed subspace $V(\varphi)$ they span, i.e., we have the identity
\[
  \sum_{k\in \mathbb{Z}}\left\vert \int_{\mathbb{R}} 
  \overline{\varphi}(x-k)f(x)dx\right\vert^{2} 
  =\int_{\mathbb{R}}|f(x)|^{2}dx \quad \mbox{for all $f \in V(\varphi)$}.
\]
\end{proposition}
\begin{proof}
  See \cite{Jor06a}.
\end{proof}

\begin{definition}
\label{D:STO.3}
Functions $W$ on $(-\pi, \pi]$ arising as in (\ref{E:STO.18}) for a system 
of wavelet coefficients $(a_{k})_{k\in \mathbb{Z}}$ (\ref{E:STO.15}), are 
called \textit{wavelet filters}.  A wavlet filter $W$ is said to be 
\textit{low-pass} if $\mu_{0}:=\delta_{0}$, i.e., the Dirac measure at 
$\theta=0$, is a Perron-Frobenius measure for $T_{W}$.

In general, if $W$ is a Lipschitz function, it is known that $T_{W}$ has a 
Perron-Frobenius measure \cite{BrJo02}.
\end{definition}

\begin{example}
\label{Ex:STO.1}
\cite{DuJo06c} Set
\begin{equation}
\label{E:STO.20} 
  W_{F}(z):=\frac{1}{6}|1+z^{2}| \quad \mbox{for $z=e^{i\theta}$}.
\end{equation}
Then $W_{F}$ is a wavelet-filter under scaling by $3$, but it is \textit{not} 
a low-pass filter.

Indeed, the following scaling law holds for $W_{F}$:
\[
  \sum_{w^{3}=z}W_{F}(w)=1, \quad \forall z=e^{i\theta} \in \mathbb{T}^{1}.
\]
We say that $W_{F}$ has scaling degree $3$.

It is proved in \cite{DuJo06c} that $W_{F}$ induces a wavelet representation 
on an $L^{2}$-space built from the middle-third-Canter construction, 
``Cantor-dust" $CD_{3}$ in $\mathbb{R}$ with Hausdorff measure 
$\mathcal{H}^{\alpha}$, $\alpha=\frac{\ln 2}{\ln 3}$, i.e., on $L^{2}($Cantor 
dust, $\mathcal{H}^{\alpha})$.

\subsubsection*{Cantor Dust $CD_{3}$}
The points $x\in CD_{3}\subset \mathbb{R}$ are encoded by 
\[
  x=a_{-k}3^{k}+a_{-k+1}3^{k-1}+\cdots +a_{0}+
  \sum_{i=0}^{\infty}\frac{a_{i}}{3^{i}}
\] 
where $k$ varies in $\mathbb{N}_{0}$, and where $a_{j}\in \{0,1,2\}$ for
$j \in \mathbb{Z}$ such that $-k \leq j$; but where $a_{j}$ attains the value 
$1$ only for at most a finite number of places.

The Perron-Frobenius measure $\mu_{0}$ for $T_{W_{F}}$ is singular with 
support $(\mu_{0})=\mathbb{T}$.
\end{example}

\section{Reprocity Rule for the Spectrum}
\label{sec:2}
In the previous section we saw that a wide class of processes are governed 
by a transfer operator $T$.  If the process in question takes places on a 
graph $G=(G^{0}, G^{1})$ with conductance $c$, then harmonic analysis on 
$G$ is phrased in terms of a Laplace operator $\Delta_{c}$ as follows:
\[
  (\Delta_{c}f)(x)=\sum_{y \sim x}c(x,y)(f(x)-f(y)), \quad 
  \mbox{for $x \in G^{0}$.}
\]

\begin{lemma}
\label{L:2.1.1}
Let $(G,c)$ and $\Delta_{c}$ be as above.  Set $p(x,y)=\frac{c(x,y)}{c(x)}$ 
for $(x,y)\in G^{1}$ and let 
\[
  (T_{c}f)(x)=\sum_{y \sim x}p(x,y)f(y),
\]
then
\[
  (\Delta_{c}f)(x)=c(x)\{f(x)-(T_{c}f)(x)\}.
\]
And conversely,
\[
  (T_{c}f)(x)=f(x)-\frac{1}{c(x)}(\Delta_{c}f)(x).
\]
\end{lemma}
\begin{proof}
Left to the reader.  
\end{proof}

Because of reference to harmonic analysis, we present 
the results in this section in terms of $\Delta_{c}$, but the lemma makes a 
translation between $\Delta_{c}$ and $T_{c}$ immediate: For example, a 
function $f$ on $G^{0}$ satisfies $\Delta_{c}f=0$ if and only if $T_{c}f=f$. 
Solution $f$ to either one of these equations are called \textit{harmonic}.

\begin{definition}
\label{D:2.1}
Let $\mathcal{H}$ be a Hilbert space, and $\mathcal{D}$ a dense linear 
subspace.  An operator $\Delta$ defined on $\mathcal{D}$ is said to be 
\textit{formally selfadjoint} if and only if 
\[
  \langle \Delta u,v \rangle=\langle u, \Delta v \rangle
\]
holds for all $u,v \in \mathcal{D}$.
\end{definition}

A further advantage of $\Delta_{c}$ over $T_{c}$ is that $\Delta_{c}$ is 
formally selftadjoint, (while $T_{c}$ is not!).

When we say that $\Delta_{c}$ is formally selfadjoint, this applies to 
either one of the two Hilbert spaces $l^{2}(G^{0})$, and $\mathcal{H}_{E}:=$ 
the energy Hilbert space.

In the case of $\mathcal{H}_{E}$, we take for $\mathcal{D}$ the linear 
span of the family $\{v_{x}|x\in G^{0}\}\subset \mathcal{H}_{E}$; see Lemma 
\ref{L:2.1} and \ref{L:2.2}.

We continue the setup from the previous section: $G=(G^{0}, G^{1})$ a fixed 
graph with vertices $G^{0}$ and edges $G^{1}$.  Let 
$c : G^{1} \to \mathbb{R}_{+}$ be a fixed conductance function.  Let 
$\Delta = \Delta_{c}$ be the Laplace operator. 
Fix an origin $o$ in $G^{0}$, and let $\{v_{x}\}_{x \in G^{0} \setminus (0)}$ 
be the system of dipoles.

\begin{lemma}
\label{L:2.1}
\cite{JoPe08} (Reproducing Kernel) The system 
$\{v_{x}\}_{x \in G^{0} \setminus (o)}$ forms a reproducing kernel in the 
sense:
\begin{equation}
\label{E:2.1}
\langle v_{x}, f\rangle_{E} = f(x)-f(o) \quad \mbox{ for all $f \in 
\mathcal{H}_{E}$,}
\end{equation}
where $\mathcal{H}_{E}$ is the energy Hilbert space.
\end{lemma}
\begin{proof}
The existence of $\{v_{x}\}$ is established with an application of Riesz's 
lemma: If $x \in G^{0}$, there is a path 
$\gamma(x)=x_{0} \to x_{1} \to \cdots \to x_{n}$, $e_{i}=(x_{i}$ $x_{i+1}) 
\in G^{1}$, (generally not unique) such that $x_{0}=0$ and $x_{n}=x$.

By Cauchy-Schwarz, we get
\begin{equation}
\label{E:2.2}
|f(x)-f(o)|^{2} \leq \sum_{i} \frac{1}{c(e_{i})}\|f\|_{E}^{2}.
\end{equation}
Riesz's lemma applied to $\mathcal{H}_{E}$, then yields  
$\exists v_{x} \in \mathcal{H}_{E}$ such that (\ref{E:2.1}) is satisfied.

We claim that $v_{x}$ satisfies the dipole equation
\begin{equation}
\label{E:2.3}
\Delta v_{x}= \delta_{x} - \delta_{o}, \quad  x \in G^{0} \setminus (o).
\end{equation}

This implies (\ref{E:2.3}), and if $\Delta h = 0$, then $w_{x}:=v_{x}+h$ solves
(\ref{E:2.3}) as well; and vice versa.
\end{proof}

\begin{lemma}
\label{L:2.2}
\cite{JoPe08} Let 
$\mathcal{D}_{0}:=\text{span}_{\mathbb{C}}\{\delta_{x}\}_{x \in G^{0}}$, and 
$\mathcal{D}_{E}:=\text{span}_{\mathbb{C}}\{v_{x}\}_{x \in G^{0} 
\setminus (o)}$. 
\end{lemma}
By ``span" we mean finite complex linear combinations, so we consider all 
finite summations
\begin{equation}
\label{E:2.4}
\mathcal{D}_{0} = \{\sum_{x}a_{x}\delta_{x}\}, \quad \mbox{ and } 
\quad \mathcal{D}_{E}=\{\sum_{x}b_{x}v_{x}\},
\end{equation}
where $\{a_{x}\}$ and $\{b_{x}\}$ denote finite systems of scalars, $a_{x}$, 
$b_{x} \in \mathbb{C}$.

Then $\Delta$ yields a density defined hermitian (i.e., formally selfadjoint) 
operator in each of the Hilbert spaces $l^{2}(G^{0})$ and $\mathcal{H}_{E}$.

Specifically, $\mathcal{D}_{0}$ is dense in $l^{2}(G^{0})$ and 
\begin{equation}
\label{E:2.5}
\langle u, \Delta v \rangle_{l^{2}} = \langle \Delta u, v \rangle_{l^{2}}, 
\quad \forall u, v \in \mathcal{D}_{0}.
\end{equation}

Moreover, $\mathcal{V}$ is dense in $\mathcal{H}_{E}$, and 
\begin{equation}
\label{E:2.6}
\langle u, \Delta v \rangle_{E} = \langle \Delta u, v \rangle_{E}, \quad
\forall u, v \in \mathcal{D}_{E}
\end{equation}
\begin{proof}
The symmetry property (\ref{E:2.5}) is immediate from the definition of 
$\Delta$.
 
We now prove (\ref{E:2.6}): Since both sides in (\ref{E:2.5}) are 
sesquilinear, it is enough, by (\ref{E:2.4}), to prove
\begin{equation}
\label{E:2.7}
\langle v_{x}, \Delta v_{y} \rangle_{E} = 
\langle \Delta v_{x}, v_{y} \rangle_{E}, \quad \forall x, y 
\in G^{0} \setminus (o). 
\end{equation}

We have 
\begin{align*}
  \langle v_{x}, \Delta v_{y} \rangle_{E} 
  &\underset{\text{by } (\ref{E:2.3})}{=} 
  \langle v_{x}, \delta_{y} - \delta_{0} \rangle_{E} \\
  &\underset{\text{by } (\ref{E:2.1})}{=} 
  (\delta_{y} - \delta_{0})(x) - (\delta_{y} - \delta_{0})(o) \\
  &= \delta_{x}(y) + 1 \\
  &\underset{\text{by symmetry}}{=} 
  \langle\delta_{x} - \delta_{0},  v_{y} \rangle_{E} \\
  &\underset{\text{by } (\ref{E:2.3})}{=} 
  \langle \Delta v_{x}, v_{y} \rangle_{E}
\end{align*}
which is the desired eq. (\ref{E:2.7}).
\end{proof}

\subsection{Two Hilbert Spaces}
\label{sec:2.1}
Let $G=(G^{0},G)$ be as above; and let $c:G^{1} \to \mathbb{R}_{+}$ be a 
fixed conductance function.  Let $\Delta$ and $T$ be the corresponding 
operators, $\Delta=\Delta_{c}$ the Laplace operator, and
\begin{equation}
\label{E:2.7.1}
  (Tf)(x)=f(x)-\frac{1}{c(x)}(\Delta f)(x), \quad x \in G^{0}.
\end{equation}
Pick a fixed $o \in G^{0}$, and let $(v_{x})_{x\in G^{0}\setminus (o)}$ be 
the corresponding reproducing kernet.

It is important to understand the two operators in the two Hilbert spaces 
$l^{2}(G^{0})$ and $\mathcal{H}_{E}$.  By (\ref{E:2.7.1}), it is enough to 
consider just $\Delta$.

As an operator in $l^{2}(G^{0})$, the operator $\Delta$ has as its domain
\begin{align*}
  \mathcal{D}_{0}&:= \text{ all finite linear combinations of } 
    \{\delta_{x}\}_{x \in G^{0}} \\
    &= \text{ span } \{\delta_{x}\}_{x \in G^{0}};
\end{align*}
while the domain in $\mathcal{H}_{E}$ is
\[
  \mathcal{D}_{E}:=\text{ span }\{ v_{x}|x\in G^{0}\setminus (o) \}.
\]

\begin{theorem}
\label{T:2.4.1}
\begin{enumerate}[(a) ]
  \item The domains in $l^{2}$ and in $\mathcal{H}_{E}$:
    \begin{enumerate}[(i) ]
      \item $\mathcal{D}_{0}$ is a dense subspace in $l^{2}(G^{0})$; and 
      \item $\mathcal{D}_{E}$ is a dense subspace in $\mathcal{H}_{E}$. 
      \item If $\sharp Nbh(x)<\infty$ for all $x\in G^{0}$, then $\Delta$ 
        maps $\mathcal{D}_{0}$ into itself; and $\Delta_{E}$ maps 
        $\mathcal{D}_{E}$ into itself.
    \end{enumerate}
  \item For all vectors $\varphi, \psi \in \mathcal{D}_{0}$, we have:
    \begin{enumerate}[(i) ]
      \item
        \[  
          \langle \varphi, \Delta \varphi\rangle_{l^{2}}
          =\sum_{x\in G^{0}}c(x)|\varphi(x)|^{2}-
           \underset{x \sim y}{\sum_{x}\sum_{y}}c(x,y)\overline{\varphi}(x)
           \varphi(y);
        \]
      \item $\langle \varphi, \Delta \varphi\rangle_{l^{2}} \geq 0$; and
      \item $\langle \varphi, \Delta \psi\rangle_{l^{2}}
        =\langle \Delta \varphi, \psi\rangle_{l^{2}}$.
    \end{enumerate}
  \item For all vectors $\varphi, \psi \in \mathcal{D}_{E}$, we have:
    \begin{enumerate}[(i) ]
      \item
        \[  
          \langle \varphi, \Delta \varphi\rangle_{\mathcal{H}_{E}}=
          \sum_{x\in G^{0}\setminus (o)}|(\Delta \varphi)(x)|^{2} +
          \left\vert\sum_{x\in G^{0}\setminus (o)}(\Delta \varphi)(x)
          \right\vert^{2};
        \]
      \item $\langle \varphi, \Delta \varphi\rangle_{\mathcal{H}_{E}} \geq 0$; 
        and
      \item $\langle \varphi, \Delta \psi\rangle_{\mathcal{H}_{E}}
        =\langle \Delta \varphi, \psi\rangle_{\mathcal{H}_{E}}$.
    \end{enumerate}
\end{enumerate}
\end{theorem}
\begin{proof}
The proof of (b)(ii) is a sequence of steps with repeated application of 
Cauchy-Schwarz's inequality.  The proof of (a)(i) is an application of the 
last equation in the proof of Lemma \ref{L:2.2}.
\end{proof}

\begin{remark}
\label{R:2.1}
The operator $\Delta_{l^{2}}$ in $l^{2}$, or $\Delta_{E}$ in 
$\mathcal{H}_{E}$, may be bounded or unbounded.  In all cases 
$\Delta_{l^{2}}$ is essentially selfadjoint in $l^{2}$ \cite{JoPe08}; but 
$\Delta_{E}$ may have defect-subspaces.
\end{remark}

\subsection{Dichotomy} 
\label{sec:2.2}

\begin{remark}
\label{R:2.2}
\cite{JoPe08}
For the graph system $(G, c)=(tree, \Iidentity)$ the Laplace 
operator 
$(\Delta, \mathcal{D}_{0})$ is bounded and selfadjoint in $l^{2}(G^{0})$.
For the energy Hilbert space $\mathcal{H}_{E}(tree)$, 
$(\Delta, \mathcal{D}_{E})$ is an \textit{unbounded} Hermitian operator.  
In fact, $\Delta$ is \textit{not} essentially selfadjoint on $\Delta$; i.e., 
$(\Delta, \mathcal{D}_{E})$ has a infinite family of distinct selfadjoint 
extensions in the Hilbert space $\mathcal{H}_{E}$.
\end{remark}

\begin{lemma}
\label{L:2.3}
Let $\mathcal{H} \langle \cdot, \cdot \rangle$ be a complex Hilbert space, and 
let $\mathcal{D}$ be a dense linear subspace in $\mathcal{H}$.

Let $L$ be a closed Hermitian operator defined on $\mathcal{D}$, i.e., $L$ is 
linear and satisfies
\begin{equation}
\label{E:2.8}
  \langle u, Lv \rangle = \langle Lu, v \rangle \quad \forall 
  u,v \in \mathcal{D}.
\end{equation}

Then the spectrum of $\Delta$ is the closure of the set
\begin{equation}
\label{E:2.9}
  NS(L):= \left\{ \frac{\langle u, Lu \rangle}{\|u\|^{2}} \biggm| u \in 
  \mathcal{D} \setminus(o) \right\}.
\end{equation}
\end{lemma}
\begin{proof}
The Hermitian property (\ref{E:2.8}) implies that the spectrum of $L$ is 
contained in $\mathbb{R}$.  

Now suppose $\lambda_{0} \in \mathbb{R}$, and that
\begin{equation}
\label{E:2.10}
  \text{dist}(\lambda_{0}, NS(L))= \epsilon_{1} > 0.
\end{equation}
We will show that $\lambda$ must then be in 
\begin{align*}
  \mathbb{R} \setminus spec(L) 
  &= \text{ the complement of the spectrum} \\
  &= \text{ the resolvent set.}
\end{align*}

Let $u \in \mathcal{D} \setminus (o)$.  Then
\[
  \| \lambda_{0}u-Lu\|^{2} 
  = \lambda_{0}^{2}\|u\|^{2}-2\lambda_{0}\langle u, Lu \rangle + \|Lu\|^{2}
\]
Setting $x_{1}:=\frac{\langle u, Lu \rangle}{\|u\|^{2}} \in NS(L)$, we get
\begin{align}
  \|\lambda_{0}u-Lu \|^{2} &= \|u\|^{2} \cdot (\lambda_{0}-x_{1})^{2}
  -\|u\|^{2}x_{1}^{2}+\|Lu\|^{2} \\
  &\underset{by (\ref{E:2.10})}{\geq}\|u\|^{2} \cdot \epsilon_{1}^{2}
  +\|Lu\|^{2}-\frac{\langle u, Lu \rangle^{2}}{\|u\|^{2}} \\
  &\geq \|u\|^{2} \cdot \epsilon_{1}^{2}
\end{align}
where we used Schwarz' inequality in the last step; viz.,
\[
  \langle u, Lu \rangle^{2} \leq \|u\|^{2} \cdot \|Lu\|^{2};
\]
or
\[
  \|Lu\|^{2}-\frac{\langle u, Lu \rangle^{2}}{\|u\|^{2}} \geq 0.
\]

By virtue of the inequality (\textit{2.11}), we may define an operator
\[
  R_{0}=R(\lambda_{0}):range(\lambda_{0}I-L) \longrightarrow \mathcal{H}
\]
by 
\begin{equation}
\label{E:2.14}
R_{0}(\lambda u - Lu) = u.
\end{equation}
Extend $R_{0}$ by setting it $=0$ on the ortho-complement 
\begin{equation}
\label{E:2.15}
  (range(\lambda_{0}I-L))^{\perp} = N(\lambda_{0}-L^{*}).
\end{equation}
Here $L^{*}$ denotes the adjoint operator.

From (\ref{E:2.14}), we calculate that $R_{0}:\mathcal{H} \to \mathcal{H}$ 
defines a bounded inverse to $\lambda_{0}I-L$, and so 
$\lambda_{0} \in \text{resolvent}(L)$; and conversely.
\end{proof}

Let $\{v_{x}\}_{x \in G^{0} \setminus(0)}$ be the system of dipoles, and set
\begin{equation}
\label{E:2.16}
M:=(\langle v_{x}, v_{y}\rangle_{E})
\end{equation}
viewed as a Hermitian matrix, $x=$ row index, $y=$ column index.

If $\xi=(\xi_{x}) \in \mathcal{F} \subset l^{2}(G^{0})$, set
\begin{equation}
\label{E:2.17}
(M\xi)_{x}=\sum_{y}M_{x, y}\xi_{y},
\end{equation}
matrix multiplication, where 
\[
  M_{x,y} := \langle v_{x}, v_{y}\rangle_{E}.
\]

Then $M$ is a density defined Hermitian operator in $l^{2}(G^{0})$.

\begin{theorem}
\label{T:2.5}
Let $(G, \mu)$ be given and let $\Delta$ be the corresponding density
defined Hermitian operator in $\mathcal{H}_{E}$.  Then
\begin{equation}
\label{E:2.18}
\text{spec}_{\mathcal{H}_{E}}(\Delta) \subset [0, \infty)
\end{equation}
and
\begin{equation}
\label{E:2.19}
\text{spec}_{\mathcal{H}_{E}}(\Delta) = (\text{spec}_{l^{2}}(M))^{-1}
\end{equation}
where we use the charactors $\frac{1}{0}=\infty$, and $\frac{1}{\infty}=0$.

Moreover,
\begin{equation}
\label{E:2.20}
(\text{spec}_{l^{2}}(M))^{-1}
 = \{1/\lambda \vert \lambda \in \text{spec}_{l^{2}}(M)\}.
\end{equation}
\end{theorem}
\begin{proof}
For $(\xi_{x}) \in \mathcal{F}$, set
\begin{equation}
\label{E:2.21}
u:=\sum_{x \in G^{0} \setminus(0)} \xi_{x}v_{x}.  
\end{equation}
Then $u \in \mathcal{V}$, and
\begin{align}
  \langle u, \Delta u \rangle_{\mathcal{H}_{E}} 
  &= \sum_{x}\sum_{y}\overline{\xi}_{x}\xi_{y}\langle v_{x}, 
     \Delta v_{y} \rangle_{E}  \\
  &= \sum_{x}\sum_{y}\overline{\xi}_{x}\xi_{y}(\delta_{x}(y)+1) \\
  &= \sum_{x}|\xi_{x}|^{2}+\left|\sum_{x}\xi_{x}\right|^{2} \geq 0
\end{align}
Since vectors in $\mathcal{H}_{E}$ are equivalence classes modulo the constant 
function on $G^{0}$, we may add the restriction $\sum_{x} \xi_{x} = 0$ in
(\ref{E:2.21}), and the operator $\Delta$ will be unchanged.

The modified equation (\textit{2.22}) then needs
\begin{equation}
\label{E:2.25}
\langle u, \Delta u\rangle_{E} = \|\xi\|_{2}^{2}
\end{equation}

  \begin{claim}
  \label{C:2.6}
  \begin{equation}
  \label{E:2.26}
  \|u\|_{\mathcal{H}_{E}}^{2} = \langle \xi, M\xi \rangle_{l^{2}}.
  \end{equation}
  \end{claim}
  \begin{proof}
  (of Claim 2.6).  We compute:
  \begin{align*}
    \|u\|_{E}^{2} &= \langle u, u\rangle_{E} \\
    &= \sum_{x}\sum_{y} \overline{\xi}_{x}\xi_{y}\langle v_{x}, v_{y} \rangle \\
    &\underset{by \ref{E:2.21}}{=} \sum_{x}\overline{\xi}_{x}(M\xi)_{x} \\
    &= \langle \xi, M\xi \rangle_{l^{2}},
  \end{align*}
  as claimed.
  \end{proof}

The desired conclusion (\ref{E:2.19}) now follows: If 
$u \in \mathcal{V} \setminus (o)$ is given by (\ref{E:2.21}), then
\begin{equation}
\label{E:2.27}
\frac{\langle u, \Delta u\rangle_{E}}{\|u\|_{E}^{2}} 
  = \frac{\|\xi\|_{2}^{2}}{\langle \xi, M\xi \rangle}.
\end{equation}
By taking closure, we obtain the sets on the two sides in (\ref{E:2.19})
\end{proof}

\begin{corollary}
\label{C:2.7}
If $\xi = (\xi_{x}) \in \mathcal{F}(G^{0} \setminus (o))$, then the 
representation 
\begin{equation}
\label{E:2.28}
u = \sum_{x} \xi_{x} v_{x}
\end{equation}
is unique; in particular, the system $(v_{x})_{x \in G^{0} \setminus (o)}$ 
is linearly independent.
\end{corollary}
\begin{proof}
Let $u \in \mathcal{V}$ have a representation (\ref{E:2.28}) as a finite 
summation with $\xi_{x} \in \mathbb{C}$.

Let $y \in G^{0} \setminus (o)$.  Then
\begin{align*} 
  \langle \delta_{y}, u \rangle_{E} 
  &= \sum_{x} \xi_{x}\langle \delta_{y}, v_{x} \rangle_{E} \\
  &\underset{by (\ref{E:2.1})}{=} 
    \sum_{x} \xi_{x}(\delta_{y}(x) - \delta_{y}(o))  \\
  &= \xi_{y}.
\end{align*}
In particular, if $u=0$, then $\xi_{y}=0$, $\forall y \in G^{0} \setminus (o)$.
\end{proof}

\begin{corollary}
\label{C:2.8}
If $F \subset G^{0} \setminus (o)$ is a finite subset, then $0$ is not in 
the spectrum of the matrix
\begin{equation}
\label{E:2.29}
M_{F} := (\langle v_{x}, v_{y}\rangle_{E})_{x,y \in F}.
\end{equation}
\end{corollary}

Suppose $o \in \text{spec}(M_{F})$ where $F$ is a fixed as in the statement 
of the Corollary \ref{C:2.8}.  Then 
\[
  \exists \xi \in l^{2}(F \setminus (o))
\]
such that 
\begin{equation}
\label{E:2.30}
(M\xi)_{x} = \sum_{y \in F} \langle v_{x}, v_{y} \rangle_{E}\xi_{y} = 0.
\end{equation}
Setting $u := \sum_{y \in F} \xi_{y}v_{y}$ we note that 
\begin{equation}
\label{E:2.31}
u \in (\{v_{x}\}_{x \in F})^{\perp}
\end{equation}

\begin{claim}
\label{C:2.9}
\begin{equation}
\label{E:2.32}
u \in (\{v_{x}\}_{x \in G^{0}\setminus(o)})^{\perp}.
\end{equation}
We need to prove this only if $x \in G^{0}\setminus F$.

Combining (\ref{E:2.26}) and (\ref{E:2.30}), we get
\begin{align*}
  \|u\|_{E}^{2} &= \langle \xi, M\xi \rangle_{l^{2}} \\ 
                &= \langle \xi, 0 \rangle_{l^{2}} \\
                &= 0,
\end{align*}
so $u=$ a constant function on $G^{0}$, and (\ref{E:2.32}) is satisfied.
\end{claim}

\section{The Energy-Inner Product}
\label{sec:3}

$(G, c)=(\text{tree} T, 1)$, $o= \emptyset$, $c \equiv 1$.  Explicitly 
form for $v_{x}$, $x \in G^{0} \setminus(o)$.
Set
\begin{equation}
\label{E:3.1}
x = (a_{1}a_{2}a_{3} \cdots a_{n}) \in G^{0} \setminus (o) \quad
a_{i} \in A={0,1}.
\end{equation}
\begin{equation}
\label{E:3.2}
\gamma(x)=\{(o a_{1}), (a_{1}a_{2}), (a_{2}a_{3}), \cdots, 
     (a_{n-2}a_{n-1}), (a_{n-1}a_{n})\}
\end{equation}
where $\gamma(x)$ is a path.
Note $\gamma(x) \subset G^{1}=$ edges in $T$.

\begin{example}
\label{Ex:3.1}
$x=101$ vertex, $\{(\varphi, 1), (1, 10), (10, 101)\}=\gamma(x)$ 
  $\sharp \gamma (x)=3$.

\end{example}

\begin{theorem}
\label{T:3.2}
Let $(T, 1)$ be as usual, $o = \emptyset$, and let $\mathcal{H}_{E}=$ the 
0 
energy span
\begin{equation}
\label{E:3.3}
\|f\|_{E}^{2} = \frac{1}{2}\underset{x \sim y}
{\sum_{x}\sum_{y}}(f(x)-f(y))^{2}
\end{equation}
but with edges $(\underline{x}, \underline{xb})=e$, $x\in G^{0}$, 
$b \in A=\{0,1\}$, $c(e) \equiv 1$.  Then the function
\begin{equation}
\label{E:3.4}
v_{x}(y):=\sharp (\gamma(x) \cap \gamma(y))
\end{equation}
solves 
\begin{equation}
\label{E:3.5}
\langle v_{x}, f\rangle_{E} = f(x)-f(o), \quad \forall f \in 
\mathcal{H}_{E}
\end{equation}
\begin{equation}
\label{E:3.6}
\Delta v_{x} = \delta_{x} - \delta_{o}, \quad x \in G^{0} \setminus 
(o)
\end{equation}
and
\begin{equation}
\label{E:3.7}
\langle v_{x}, v_{y}\rangle_{E} = \sharp (\gamma(x) \cap \gamma(y)) 
\quad \forall x, y \in G^{0} \setminus (o)
\end{equation}
\end{theorem}
\begin{proof}
\textit{Proof of (\ref{E:3.6}).}  By (\ref{E:3.4}) 
$x=(a_{1}a_{2} \cdots a_{n}) \in G^{0} \setminus (o)$.  Let $x$ be 
as in (\ref{E:3.2}).  Set $\gamma(x)=$ RHS in (\ref{E:3.2}) $\subseteq G^{1}$.
Neighbors of 
\begin{align*}
  x &\longrightarrow a_{1}\cdots a_{n-1} \\
    &\longrightarrow x0 \\
    &\longrightarrow x1
\end{align*}
If $x=a$, $n=1$, Nbh$(x)=\{o, a0, a1\}$

\subsubsection*{Cases}
\begin{itemize}
\item[n=1] See Figure 4 
  \begin{align*}
    (\Delta v_{x})(o) &= 2v_{x}(o)-v_{x}(0)-v_{x}(1) \\
    &= 0-1 \\
    &= \delta_{x}(o)-\delta_{o}(o)
  \end{align*}
  \begin{align*}
    (\Delta v_{x})(x) &= 3v_{x}(x)-v_{x}(o)-v_{x}(x0)-v_{x}(x1) \\
    &\underset{use (\ref{E:3.4})}{=} 3-0-1-1=1 \\
    &= (\delta_{x}-\delta_{o})(x)
  \end{align*}
  Now, let $y \in G^{0} \setminus \{o, x\}$. 
  $y=(b_{1}b_{2}\cdots b_{k})$, $b_{i} \in A=\{0,1\}$. 
  Suppose $x \subseteq y$ 
  \[
    \Delta v_{x}(y) \underset{by (\ref{E:3.4})}{=} 3-1-1-1=0
  \]
  More cases are $\equiv 0$.
\item[n$>$1] $x=(a_{1}a_{2}\cdots a_{n})$. A computation yields
  \[
    \Delta_{x}(o) = 0-v_{x}(0)-v_{x}(1) = 0-1 = -1  
  \]
  \begin{align*}
    \Delta_{x}(x) &= 3n-(n-1)-2n = 1 \\
                  &= (\delta_{x}-\delta_{o})(x) 
  \end{align*}
  \[
    \Delta_{x}(y) = 0 \quad y \in G^{0} \setminus \{o, x\}  
  \]
  Several cases e.g. $y \leq x$, etc.
  \begin{align*}
    \Delta_{x}(y) &= 3v_{x}(y)-v_{x}(b_{1} \cdots b_{k-1})-v_{x}(y0)-v_{x}(y1)\\
                  &= 3k-(k-1)-(k+1)-k=0, \text{ etc.} 
  \end{align*}
  Computation of 
  \begin{align*}
    \|v_{x}\|_{E}^{2} &= E(v_{x}) \\
    &= \frac{1}{2} \underset{s \sim t}{\sum_{s} \sum_{t}} 
    c_{s,t}(v_{x}(s)-v_{x}(t))^{2}  \\
    &= \langle v_{x}, \Delta v_{x} \rangle_{l^{2}} 
      \underset{by (\ref{E:3.6})}{=} \langle v_{x}, \delta_{x} - 
      \delta_{o}\rangle_{l^{2}}  \\
    &= v_{x}(x) - v_{x}(o) = n-0 \\
    &= \sharp(\gamma(x))
  \end{align*}



  \begin{align*}
    \langle v_{x}, v_{y}\rangle_{E} 
      &\underset{\text{all finite functions, see } (\ref{E:3.4})}{=} 
      \langle v_{x}, \Delta v_{y}\rangle_{l^{2}} \\
    &\underset{by (\ref{E:3.6})}{=} \langle v_{x}, \delta_{y}- 
      \delta_{o} \rangle_{l^{2}} \\
    &\underset{by (\ref{E:3.4})}{=} v_{x}(y)-v_{x}(o)=0  \\
    &\underset{by (\ref{E:3.4})}{=} \sharp(\gamma(x) \cap \gamma(y)) 
  \end{align*}
\end{itemize}
\end{proof}

Set $M=(\langle v_{x}, v_{y}\rangle_{E})=(\sharp(\gamma(x) \cap \gamma(y)))$, 
$x, y \in G^{0} \setminus (o)$.  Given 
\[
  \text{spec}_{l^{2}}(M)
  =(\text{spec}_{\mathcal{H}_{E}}(\Delta))_{\text{spec}_{\mathcal{H}_{E}(\Delta) \to \infty}}^{-1}
\]


From our theorem above $\Delta$ (unbounded spectrum), 
closure$\Delta = \mathcal{V}$, $\mathcal{V} \subset \mathcal{H}_{E}$.

\begin{corollary}
\label{C:3.3}
$\forall \epsilon \quad \exists F \subset G^{0} \setminus (o)$ finite, 
$\exists \lambda \in \text{spec}_{l^{2}}(M_{F})$ such that 
$\lambda < \epsilon$.

\end{corollary} 

Note 
\begin{align*}
  M_{F} &= (\langle v_{x}, v_{y}\rangle_{E})_{x,y \in F} \\
        &= (\sharp(\gamma(x) \cap \gamma(y)))_{x,y \in F}
\end{align*}
and $0 \notin \text{spec}_{l^{2}}(M_{F})$.

Problem: Find a systematic way of selecting $F$.  See Figure 8.

It is much easier to find $M_{F}$ with $\text{spec}_{l^{2}}M_{F} \to \infty$.

\begin{example}
\label{Ex:3.4}
\begin{equation}
\label{E:3.8}
  \begin{pmatrix}
    1 & 0 & 1 \\
    0 & 2 & 0 \\
    1 & 0 & 2 
  \end{pmatrix}
  \qquad \text{or} \qquad
  \begin{pmatrix}
    1 & 1 \\
    1 & n
  \end{pmatrix},
\end{equation}
\[
  \lambda_{n}^{\pm} = \frac{n+1 \pm \sqrt{(n+1)^{2}-4(n-1)}}{2}, \quad 
  \mbox{and}
\]
\begin{align*}
  \lambda_{n}^{-} &= \frac{n+1-\sqrt{(n+1)^{2}-4(n-1)}}{2} \\
                  &= \frac{2(n-1)}{n+1-\sqrt{n^{2}-2n+2}} \\
                  &\underset{n \to \infty}{\longrightarrow} 1
\end{align*}
Actually both expand part of $\text{spec}_{l^{2}}M$ as intervals.
\end{example}


\section{Karhunen-Lo\`{e}ve}
\label{sec:4}

\begin{definition}
\label{D:4.1}
$F \subset G^{0}\setminus (o)$ $F$ finite, $G^{0}$ infinite, $(G, c)$ fixed. 
Fix $o \in G^{0} \rightsquigarrow \Delta = \Delta_{c}$, $\mathcal{H}_{E}$  
energy Hilbert space $\langle v_{x}, f\rangle_{E}=f(x)-f(o)$, 
$\forall f \in \mathcal{H}_{E}$.
\end{definition}

\begin{definition}
\label{D:4.2}
$\mathcal{H}_{E}(F):=\overline{span}_{x\in F}\{v_{x}\}$
\end{definition}

General $(G, c) \to \text{Fix point } o \in G^{0} \to \Delta_{c}$, 
$\mathcal{H}_{E}$, $G^{0}$ infinite.  Fix $v_{x}$, for 
$x \in G^{0} \setminus (o)$,
determined from Riesz applied to $\mathcal{H}_{E}$.
\begin{equation}
\label{E:4.1}
  \langle v_{x}, f\rangle_{E} = f(x)-f(o), \quad x \in G^{0} \setminus (o),
\end{equation}
and consider the infinite matrix
\begin{equation}
\label{E:4.2}
  M= (\langle v_{x}, v_{y}\rangle_{E}), \quad x,y \in G^{0} \setminus (o)
\end{equation}
and its finite $F \times F$ submatrices
\begin{equation}
\label{E:4.3}
  M_{F}= (\langle v_{x}, v_{y}\rangle_{E}), \quad x,y \in F 
\end{equation}
so the matrices are $\infty \times \infty$, or $|F| \times |F|$.

\subsection*{Important Formula}
Observe
\begin{align}
   \langle v_{x}, v_{y}\rangle_{E} &= v_{y}(x)- v_{y}(o) \\
                                   &= v_{y}(x) = v_{x}(y);
\end{align}
in other words $k_{E}(x,y):=v_{x}(y)$ is a reproducing kernel.

Since $x, y \in G^{0} \setminus (o)$; and $v_{x}:G^{0}\to \mathbb{R}$ 
(i.e., real valued) convention: $v_{x}(o)=0$.
Diagonalization motivated by the classical Karhunen-Lo\`{e}ve theorem, see 
\cite{JoSo07} and \cite{Loe55}. 

\subsection{Finite-dimensional Approximation}
Apply the Spectral Theorem to $M$ and $M_{F}$.  The Hilbert space is 
$l^{2}(G^{0})$ or $l^{2}(F) \simeq  \mathbb{C}^{|F|}$ with  
$\langle \xi, \eta\rangle_{2}=\sum_{x}\overline{\xi_{x}}\eta_{x}$ as inner product.

For $(M_{F}, l^{2}(F))$ the spectrum is always discrete, and for some cases 
i.e., $(M, l^{2}(G^{0}))$ it may not be discrete.

In the discrete case, there exists $M=Mf$ ONB 
$\xi_{1}, \xi_{2}, \ldots \in l^{2}(G^{0})$ or $l^{2}(F)$ eigenvectors

\begin{equation}
\label{E:4.6}
  \langle \xi_{j}, \xi_{k}\rangle_{2}=\sum_{x}\overline{\xi_{j}(x)}\xi_{k}(x)
  = \delta_{j, k} =
  \begin{cases}
    0	&\text{if $j=k$} \\
    1	&\text{if $j\neq k$} 
  \end{cases}
\end{equation}
$\xi = \xi^{F} \in l^{2}(F)$ such that 
\begin{equation}
\label{E:4.7}
  M^{F}\xi_{j} = \lambda_{j}\xi_{j},  
    \quad \lambda_{1} \geq \lambda_{2} \geq \cdots >0 
    \quad \xi_{j} \in l^{2}(F), \|\xi_{j}\|_{2}=1 
    \quad \mbox{(in the $F$-case)}
\end{equation}
In the infinite case spec$(M)$ for $l^{2}(G^{0})$ may accumulate both at $0$ 
and at $\infty$.

Since $M_{xy}^{F}=\langle v_{x}, v_{y} \rangle_{E\Lambda} \in \mathbb{R}$, 
we may take all $\xi_{k}:G^{0} \to \mathbb{R}$ real valued.  Fix 
$F \subset G^{0} \setminus (0)$: $\xi_{k}^{F} \in l^{2}(F)$.  Set 
\[
  w_{k}^{F}(\cdot)=\frac{1}{\lambda_{k}} \sum_{x\in F}\xi_{k}^{F}(x)v_{x}(\cdot)
\]
i.e., 
\begin{equation}
\label{E:4.8}
  w_{k}^{F}(z)=\frac{1}{\lambda_{k}} \sum_{x\in F}\xi_{k}^{F}(x)v_{x}(z), 
  \quad \forall z \in G^{0}
\end{equation}

\begin{lemma}
\label{L:4.1}
  If $F$ is fixed then $\xi_{k}^{F}\in l^{2}(F)$ is an ONB. Set 
  \begin{equation}
  \label{E:4.9}
    M_{F}\xi_{k}^{F}=\lambda_{k}^{F}\xi_{k}^{F},
  \end{equation}
  then
  \[
    w_{k}^{F}:G^{0} \to \mathbb{R}, \quad w_{k}^{F} \in \mathcal{H}_{E}
  \]
  is an extension of $\xi_{k}^{F}:F \to \mathbb{R}$ from $F$ to $G^{0}$.
\end{lemma}
\begin{proof}
  By (\ref{E:4.8}) if $z \in F$:
  \begin{align*}
    w_{k}^{F}(z)&=\frac{1}{\lambda_{k}} 
                   \sum_{x\in F}v_{x}(z)\xi_{k}^{F}(x) \\
                   &\underset{\text{by} (4.4)}=\frac{1}{\lambda_{k}} 
                   \sum_{x\in F}M_{z,x}\xi_{k}^{F}(x) \\
                   &\underset{\text{by} (\ref{E:4.9})}= \frac{1}{\lambda_{k}}
                   (M_{F}\xi_{k}^{F})_{z} \\
                   &=\frac{\lambda_{k}}{\lambda_{k}}\xi_{k}^{F}(z) \\
                   &= \xi_{k}^{F}(z)
  \end{align*}
\end{proof}

\begin{lemma}
\label{L:4.2}
  Fix $F \subset G^{0}\setminus(0)$ finite, and let
  \begin{equation}
  \label{E:4.10}
    w_{k}^{F}(\cdot)=\frac{1}{\lambda_{k}} 
      \sum_{x\in F}\xi_{k}^{F}(x)v_{x}(\cdot), \quad k \in (1, 2, \cdots, |F|)
      \quad \mbox{\text{as in Lemma} \ref{L:4.1}}.
  \end{equation}
  Then $\{w_{k}^{F}\}_{k}$ is an orthonormal system in $H_{E}$ 
  (thus in each of the Hilbert spaces) i.e., with the inner product
  \begin{equation}
  \label{E:4.11}
    \langle u,v\rangle_{E}:=\frac{1}{2} 
    \sum_{\text{all } xy}\sum_{x \sim y}c_{xy}
    (\overline{u(x)}-\overline{u(y)})(\overline{v(x)}-\overline{v(y)}).
  \end{equation}
  We have
  \begin{equation}
  \label{E:4.12}
    \langle w_{j}^{F}, w_{k}^{F}\rangle_{E}=\frac{1}{\lambda_{k}}
    \delta_{j, k} =
      \begin{cases}
        \frac{1}{\lambda_{k}}	&\text{if $j=k$} \\
        0			&\text{if $j\neq k$} 
      \end{cases} 
  \end{equation}
\end{lemma}
\begin{proof}
  We have: 
  \begin{align*}
    \langle w_{j}^{F}, w_{k}^{F}\rangle_{E}
      &\underset{\text{by} (\ref{E:4.10})}=
      \frac{1}{\lambda_{j}\lambda_{k}}\sum\sum_{xy \in F}\xi_{j}^{F}(x)
      \xi_{k}^{F}(y)\langle v_{x}, v_{y}\rangle_{E} \\
      &\underset{\text{by} (\ref{E:4.3})}=\frac{1}{\lambda_{j}\lambda_{k}}
      \sum_{x \in F}\xi_{j}^{F}(x)\sum_{y \in F}M_{xy}^{F}\xi_{k}^{F}(y) \\
      &=\frac{1}{\lambda_{j}\lambda_{k}}\langle \xi_{j}^{F}, 
      M^{F}\xi_{k}^{F}\rangle_{2} \\
      &=\frac{1}{\lambda_{j}}\langle \xi_{j}^{F}, \xi_{k}^{F}\rangle_{2} \\
      &\underset{\text{by} (\ref{E:4.6})}=\frac{1}{\lambda_{j}}\delta_{j, k}=
      \begin{cases}
        \frac{1}{\lambda_{j}}	&\text{if $k=j$} \\
        0			&\text{if $k\neq j$} 
      \end{cases}. 
  \end{align*}
\end{proof}

Set $u_{j}^{F}=\sqrt{\lambda_{j}}w_{j}^{F}$; then 
\[
  \langle u_{j}^{F}, u_{k}^{F}\rangle_{E}=\delta_{j, k}, \quad j, k \in
  \{1, 2, \cdots, |F|\}.
\]

\subsection{Normalization}
\label{sec:4.1}
The following different normalization 
$u_{j}^{F}=\sqrt{\lambda_{j}}w_{j}^{F}$ satisfies 
\begin{equation}
\label{E:4.13}
  \|u_{j}^{F}\|_{H_{E}}=1,
\end{equation}
so
\begin{equation}
\label{E:4.14}
  u_{j}^{F}(\cdot)=\frac{1}{\lambda_{j}}\sum_{x \in F}\xi_{j}(x)v_{x}(\cdot).
\end{equation}
Note that the 
\begin{equation}
\label{E:4.15}
  u_{j}^{F}|_{F}=\sqrt{\lambda_{j}}\xi_{j}(\cdot) \quad \mbox{ on $F$.}
\end{equation}

\subsection{Projection Valued Measures}
Set
\begin{equation}
\label{E:4.16}
  P^{F}(\lambda_{j}):=|u_{j}^{F}><u_{j}^{F}|; \quad 
  \mbox{Dirac notation for rank-one projection,}
\end{equation}
so a projection in $H_{E}$ on the one-dimensional subspace 
$\mathbb{C}u_{j}^{F}$.  Then $P^{F}(\cdot)$ is an 
orthonormal projection system, and it has a limit as $F \to \infty$ which is 
a global spectral measure.
 
We claim that
\begin{equation}
\label{E:4.17}
  s_{\Delta}(u_{j}^{F})=\langle u_{j}^{F}, \Delta u_{j}^{F} \rangle \in
  spec_{H_{E}}(\Delta v)
\end{equation}

\begin{lemma}
\label{L:4.3}
  (Spectral Reprocity)
  \begin{equation}
  \label{E:4.18}
    s_{\Delta}(u_{j}^{F})=\frac{1}{\sqrt{\lambda_{j}}}
    \left(1+\left\vert\sum_{x \in F}\xi_{j}(x)\right\vert^{2}\right) 
  \end{equation}
\end{lemma}
\begin{proof}
  \begin{align*}
  s_{\Delta}(u_{j}^{F})&\underset{by (\ref{E:4.14})}=
  \frac{1}{\sqrt{\lambda_{j}}}\langle\left(\sum_{x \in F}\xi_{j}(x)v_{x}\right),
  \Delta\left(\sum_{y \in F}\xi_{j}(y)v_{y}\right)\rangle_{E} \\
  &=\frac{1}{\sqrt{\lambda_{j}}}\sum_{x \in F}\sum_{y \in F}\xi_{j}(x)\xi_{j}(y)
  \langle v_{x}, \Delta v_{y} \rangle_{E} \\
  &=\frac{1}{\sqrt{\lambda_{j}}}\sum_{x \in F}\sum_{y \in F}\xi_{j}(x)\xi_{j}(y)
  (\delta_{x}(y)+1) \\
  &=\frac{1}{\sqrt{\lambda_{j}}}\left(\|\xi_{j}\|_{2}^{2}+
     \left\vert\sum_{x \in F}\xi_{j}(x)\right\vert^{2}\right)  \\
  &\underset{by (\ref{E:4.13})}=\frac{1}{\sqrt{\lambda_{j}}}
    \left(1+\left\vert\sum_{x \in F}\xi_{j}(x)\right\vert^{2}\right)
  \end{align*}
\end{proof}

\begin{example}
\label{Ex:4.1}
\[
  M^{F}=   
  \begin{pmatrix}
    1 & 0 \\
    0 & 3
  \end{pmatrix}, \quad sp=\{1, 3\},
\]
same spectrum, but different $M^{F}$. 
\[
  \xi_{\lambda =1} =
  \begin{pmatrix}
    1  \\
    0 
  \end{pmatrix} 
  \quad 
  \xi_{\lambda =3} =
  \begin{pmatrix}
    0  \\
    1 
  \end{pmatrix}, 
  \quad \langle \xi_{1}\rangle = \langle \xi_{3}\rangle =1.
\]
\end{example}

Set $R_{F}(\lambda):=\frac{1}{\lambda}
(1+\left\vert \langle \xi_{\lambda}^{F}\rangle \right\vert^{2})$.  Then 
$\langle u_{\lambda}, \Delta u_{\lambda}\rangle=R_{F}(\lambda)$; see 
Lemma \ref{L:4.3}.

In the examples:
\[
  M^{F}=   
  \begin{pmatrix}
    2 & 1 \\
    1 & 2 
  \end{pmatrix}, \quad 
  \begin{cases}
    R_{F}(1)=1 \\
    R_{F}(3)=\frac{1}{3}(1+2)=1 
  \end{cases} \quad \mbox{smaller for $M^{F}$ off-diagonal.}
\]
\[
  M^{F}=   
  \begin{pmatrix}
    1 & 0 \\
    0 & 3
  \end{pmatrix}, \quad 
  \begin{cases}
    R_{F}(1)=\frac{1}{1}(1+1^{2})=2 \\
    R_{F}(3)=\frac{1}{3}(1+1^{2})=\frac{2}{3} 
  \end{cases}
\]
\[
  M^{F}=   
  \begin{pmatrix}
    3 & 3 \\
    3 & 7 
  \end{pmatrix}, \quad \lambda_{\pm}=5\pm\sqrt{13}
\]
\[
  M^{F}=   
  \begin{pmatrix}
    3 & 1 \\
    1 & 4 
  \end{pmatrix}, \quad \lambda_{\pm}=\frac{7\pm\sqrt{5}}{2}
\]
\[
  \lambda =5\pm\sqrt{13} \Rightarrow 
  R(\lambda) = \frac{1}{\lambda}+\frac{\lambda}{1+(\frac{2-\sqrt{13}}{3})^{2}}
  < \frac{1}{\lambda}+\lambda
\]
\[
  M^{F}=   
  \begin{pmatrix}
    1 & 1 \\
    1 & m 
  \end{pmatrix}, \quad m \to \infty \quad \lambda_{\pm}=\frac{m+1\pm\sqrt{(m+1)^{2}-4(m-1)}}{2}
\]

In both cases, we have:
\[
  R_{F}(\lambda)=\frac{1}{\lambda}+\frac{\lambda}{1+(\lambda -1)^{2}}
\]
\[
  R_{F}(\lambda_{-})=\frac{\lambda_{+}}{m-1}+
  \frac{\lambda_{-}}{1+(\lambda_{-} -1)^{2}} \sim 4 \mbox{ as $m\to \infty$}. 
\]


We now illustrate by an example that points in the spectrum can go into 
$\infty$:
\[
  s_{\Delta}(u)=\frac{\langle u, \Delta u\rangle}{\|u\|_{E}^{2}} \to \infty
\]
If $\lambda \in spec_{l^{2}}(M^{F})$ set 
$u_{\lambda}=\frac{1}{\sqrt{\lambda}}\sum_{x\in F}\xi_{\lambda}(x)v_{x}(\cdot)$ $M\xi_{\lambda}=\lambda \xi_{\lambda}$, $\|\xi_{\lambda}\|_{2}=1$ 
$\Rightarrow \|u_{\lambda}\|_{E}=1$ so 
$s_{\Delta}(u)=\langle u,\Delta u\rangle 
=\frac{1}{\lambda}(1+\|P_{\lambda}e\|_{2}^{2})$, $e=e_{F}=\chi_{F}(\cdot)$, 
$P_{\lambda}e=\langle \xi_{\lambda}, e\rangle_{2} \xi_{\lambda}$

\begin{theorem}
\label{T:4.4}
The truncated operators 
$P_{\mathcal{H}_{E}(F)}\Delta_{\mathcal{D}_{E}}P_{\mathcal{H}_{E}(F)}$ has 
spectral growth $\simeq \mathcal{O}(\sharp F)$; so $\Delta_{E}$ is 
unbounded in $\mathcal{H}_{E}$.
\end{theorem}
\begin{proof}
The idea is to perform a diagonalization of an infinite matrix $(M_{x,y})$ 
$x,y \in G^{0}\setminus(o)$; a method inspired by Karhunen-Lo\`{e}ve 
\cite{JoSo07, Loe55}. 
Here $F\subset G^{0}\setminus(o)$ is fixed and finite.  The following 
computations refer to $F$: $(\xi_{k})$ is an ONB in $l^{2}(F)$ satisfying
(\ref{E:4.19}) below; set 
$w_{k}=\frac{1}{\lambda_{k}}\sum_{x\in F}\xi_{k}(x)v_{x}$, and 
$v_{k}=\sqrt{\lambda_{k}}w_{k}
=\frac{1}{\sqrt{\lambda_{k}}}\sum_{x\in F}\xi_{k}(x)v_{x}$. Then
\begin{equation}
\label{E:4.19}
  M^{F}\xi_{k}=\lambda_{k}\xi_{k}, \quad \mbox{and} \quad 
  \langle \xi_{j},\xi_{k} \rangle_{l^{2}(F)}=\delta_{j,k}.
\end{equation}
We may now compute the matrices:
\begin{align*}
  \langle u_{j}, \Delta u_{k} \rangle_{E}
  &=\frac{1}{\sqrt{\lambda_{j}\lambda_{k}}}\underset{F \times F}{\sum\sum}
  \xi_{j}(x)\xi_{k}(x)\langle v_{x}, \Delta v_{y}\rangle_{E}  \\
  &=\frac{1}{\sqrt{\lambda_{j}\lambda_{k}}}\underset{F \times F}{\sum\sum}
  \xi_{j}(x)\xi_{k}(x)(\delta_{x}(y)+1).
\end{align*}
Set $\sum_{x\in F}\xi_{j}(x)=\langle \xi_{j},e \rangle_{2}
=\langle \xi_{j}\rangle$ where $e=e^{F}=\chi_{F}$.

Then the matrix entries are:
Off-Diagonal:
\[
  \langle u_{j}, \Delta u_{k} \rangle_{E}
  =\frac{1}{\sqrt{\lambda_{j}\lambda_{k}}}(\delta_{j,k}
  +\langle \xi_{j}\rangle \langle \xi_{k}\rangle); 
\] and Diagonal:
\[
  \langle u_{j}, \Delta u_{j} \rangle_{E}
  =\frac{1}{\lambda_{j}}(1+\langle \xi_{j}\rangle^{2}).
\] 

We further used the following identity:
\begin{align*}
  \underset{F \times F}{\sum\sum}\delta_{x}(y)\xi_{j}(x)\xi_{k}(y) 
  &=\langle \xi_{j}, \xi_{k} \rangle_{l^{2}}(F) \\
  &=\delta_{j,k} \quad \mbox{by (\ref{E:4.19}).}
\end{align*}

This may be summarized in the following matrix form:
\[
  \begin{pmatrix}
    \frac{1}{\lambda_{1}}(1+\langle \xi_{1}\rangle^{2}) & 
    \frac{\langle \xi_{1}\rangle\langle \xi_{2}\rangle}
    {\sqrt{\lambda_{1}\lambda_{2}}} & 
    \frac{\langle \xi_{1}\rangle\langle \xi_{3}\rangle}
    {\sqrt{\lambda_{1}\lambda_{3}}} &  \cdots \\
    \frac{\langle \xi_{1}\rangle\langle \xi_{2}\rangle}
    {\sqrt{\lambda_{1}\lambda_{2}}} & 
    \frac{1}{\lambda_{2}}(1+\langle \xi_{2}\rangle^{2}) & 
    \frac{\langle \xi_{2}\rangle\langle \xi_{3}\rangle}
    {\sqrt{\lambda_{2}\lambda_{3}}} &  \cdots \\
    \frac{\langle \xi_{1}\rangle\langle \xi_{3}\rangle}
    {\sqrt{\lambda_{1}\lambda_{3}}} & 
    \frac{\langle \xi_{2}\rangle\langle \xi_{3}\rangle}
    {\sqrt{\lambda_{2}\lambda_{3}}} &  
    \frac{1}{\lambda_{3}}(1+\langle \xi_{3}\rangle^{2}) &  \cdots \\
    \vdots & \vdots & \vdots & \ddots
  \end{pmatrix}
\]

If for some $\delta\in \mathbb{R}^{+}$, $\lambda{j} \geq \delta$, i.e., 
bounded from below, then the operator
\[
  \begin{pmatrix}
    \frac{1}{\lambda_{1}} & 0 & 0 & 0 & \cdots & 0 \\
    0 & \frac{1}{\lambda_{2}} & 0 & 0 & \cdots & 0 \\
    0 & 0 & \frac{1}{\lambda_{3}} & 0 & \cdots & 0 \\
    \vdots & \vdots & \vdots & \vdots & \ddots & 0
  \end{pmatrix}
\]
is bounded.  So
\begin{equation}
\label{E:4.20}
  \left(\frac{1}{\sqrt{\lambda_{j}\lambda_{k}}}\langle \xi_{j}\rangle 
  \langle \xi_{k}\rangle \right)
\end{equation}
must be unbounded, i.e., $\|\cdot\|_{l^{2}(F)\to l^{2}(F)} \to \infty$. But 
(\ref{E:4.20}) is a rank-one operator;
\[
  |\rho ><\rho|, \quad \rho=\rho^{F}; \quad F \subset G^{0}\setminus (o) 
  \text{ is fixed  where } \rho=(e_{j})\in l^{2}(1,2,\cdots, \sharp F),
\]
i.e., $\rho=\rho^{F}$ and 
$\rho_{j}^{F}=\frac{\langle \xi_{j}^{F}\rangle}{\sqrt{\lambda_{j}}}$, 
$\lambda_{j}=\lambda_{j}^{F}$. 

Now, 
\[
  \|\rho^{F}\|^{2}_{l^{2}(1,\cdots, \sharp F)} = 
  \sum_{j=1}^{\sharp F}\frac{\langle \xi_{j}^{F}\rangle^{2}}{\lambda_{j}}.
\]
So in conclusion
\[
  \lim_{F\to \infty}\sum_{j=1}^{\sharp F}
  \frac{\langle \xi_{j}^{F}\rangle^{2}}{\lambda_{j}(F)} = \infty.
\]

Pick $\delta \in \mathbb{R}_{+}$ and assume $\lambda_{j}^{F} \geq \delta$.  
Then we need 
\[
  \lim_{F\to \infty}\sum_{j=1}^{\sharp F}\langle \xi_{j}^{F} \rangle^{2} 
  =\infty.
\]

We have $\xi_{\lambda}^{F}(j)=\xi_{j}$, 
$M^{F}\xi_{\lambda}^{F}=\lambda_{j}^{F}\xi_{\lambda}^{F}$, 
$\|\xi_{\lambda}\|_{2}=1$, $\langle \xi_{\lambda}^{F}\rangle 
=\sum_{x\in F}\xi_{\lambda}^{F}(x)$, and 
$\sum_{\lambda}\langle \xi_{\lambda}^{F}\rangle^{2}=\sharp F$; so indeed
\[
  \lim_{F\to \infty} \sum_{\lambda}\langle \xi_{\lambda}^{F}\rangle^{2} 
  = \lim_{F}\sharp F =\infty.
\]
Conclusion: $spec_{\mathcal{H}_{E}(N)}(\Delta_{\mathcal{D}_{E}}) \sim 
(\sharp F) \to \infty$

\end{proof}

\begin{acknowledgements}
The authors are please to acknowledge helpful discussions, both recent and 
not so recent, with John Benedetto, B. Brenken, Ilwoo Cho, D. Dutkay, 
Keri Kornelson, Kathy Merrill, P. Muhly, Judy Packer, Erin Pearse, 
Steen Pedersen, Gabriel Picioroaga, Karen Shuman.
\end{acknowledgements}

\bibliographystyle{plain}
\bibliography{spectrum}

\end{document}